\documentclass[12pt]{amsart}
\usepackage{amsmath,amsthm, graphicx, amsfonts}
\usepackage{color}

\usepackage{rotating}

\usepackage{amssymb}

\usepackage{slashbox}
\DeclareMathOperator*{\argmin}{argmin}

\def\be{\begin{equation}}
\def\ee{\end{equation}}
\def\ba{\begin{array}}
\def\ea{\end{array}}
\def\bc{\begin{center}}
\def\ec{\end{center}}

\def\ZZ{\rm {{\rm Z}\kern-.48em{\rm Z}}}
\def\RR{\rm \hbox{I\kern-.2em\hbox{R}}}
\def\CC{\rm \hbox{C\kern -.5em{\raise .32ex
\hbox{$\scriptscriptstyle |$}}\kern - .22em{\raise .6ex
\hbox{$\scriptscriptstyle |$}}\kern
.4em}}
\def\CC{ {\mathbb C} }
\def\RR{ {\mathbb R} }
\def\ZZ{ {\mathbb Z} }

\newtheorem{Tm}{Theorem}[section]
\newtheorem{Df}[Tm]{Definition}

\newtheorem{pr}[Tm]{Proposition}
\newtheorem{Lm}[Tm]{Lemma}
\newtheorem{Ex}[Tm]{Example} 
\newtheorem{re}[Tm]{Remark} 
\makeatletter

\newcommand{\Rmnum}[1]{\expandafter\@slowromancap\romannumeral #1@}

\newcommand{\medcap}{\mathbin{\scalebox{1.5}{\ensuremath{\cap}}}}
\newcommand{\medcup}{\mathbin{\scalebox{1.5}{\ensuremath{\cup}}}}

\makeatother
\numberwithin{equation}{section}

\begin{document}
\bibliographystyle{plain}

\title[Spatially Distributed  Sampling and Reconstruction]{Spatially Distributed Sampling and Reconstruction}

\author{Cheng Cheng, Yingchun Jiang  and Qiyu  Sun}
\address{Cheng: Department of Mathematics, University of Central Florida, Orlando, Florida 32816, USA}
\email{cheng.cheng@knights.ucf.edu}

\address{Jiang: School of Mathematics and Computational Science, Guilin University of
 Electronic Technology, Guilin, Guangxi 541004, China.
}
\email{guilinjiang@126.com}

\address{Sun: Department of Mathematics, University of Central Florida, Orlando, Florida 32816, USA}
\email{qiyu.sun@ucf.edu}

\thanks{This project is partially supported by the National
Natural Science Foundation of China (Nos. 11201094
and 11161014), Guangxi Natural Science Foundation
(2014GXNSFBA118012), and the National Science
Foundation (DMS-1412413).}

\begin{abstract}\  A spatially distributed system contains a large amount
of agents with limited sensing, data processing, and communication
capabilities. Recent technological advances have opened
up possibilities to deploy  spatially distributed systems for signal sampling and reconstruction.
In this paper,  we introduce a graph structure for a distributed sampling and reconstruction system
by coupling agents in a spatially distributed system with innovative positions of signals.
A fundamental problem in sampling theory
is  the robustness of  signal reconstruction  in the presence of
sampling noises. For a distributed sampling and reconstruction system, the robustness could be reduced to the stability of its sensing matrix.
In a traditional centralized sampling and reconstruction system,
the stability
of the sensing matrix could be verified by its central processor,
but the above procedure is infeasible in a distributed sampling and reconstruction system as it is decentralized.
 In this paper, we
 split a distributed sampling and reconstruction system
into a family of overlapping smaller subsystems, and we show that
the stability of the sensing matrix holds if and only if its quasi-restrictions to those  subsystems
have uniform stability. This new stability criterion could be  pivotal for the design of a robust
distributed sampling and reconstruction system
against supplement, replacement and impairment of agents, as we only need to check the uniform stability of affected
subsystems. In this paper, we also propose
an exponentially convergent distributed algorithm for  signal  reconstruction,
that provides
a suboptimal approximation to the original signal in the presence of bounded sampling noises.
\end{abstract}


\keywords{Spatially distributed systems, distributed sampling and reconstruction systems,
signals on a graph, distributed algorithms,
 finite rate of innovation,  Beurling dimensions,
 localized matrices, inverse-closed subalgebras.}

 \maketitle

\section{Introduction}

Spatially distributed
systems (SDS) 
have been
 widely used in (underwater) multivehicle and  multirobot networks, wireless sensor networks,  smart grids, etc (\cite{aky02, chong2003, curtin89,  yick2008, zhaobook2004}).
Comparing with traditional centralized systems that have a powerful central processor and 
reliable communication
between agents and the central processor,
an  SDS  could give unprecedented capabilities
especially when 
 creating a data exchange network
 requires significant  efforts
(due to physical barriers such as interference),
or when establishing a centralized  processor presents the daunting
challenge of 
processing all the information
(such as  big-data problems).
 In this paper, we consider  SDSs for signal sampling and reconstruction, and we
 describe the  topology  of an SDS by an undirected (in)finite  graph
 \begin{equation}\label{g.graph}
 {\mathcal G}:=(G, S),\end{equation}
 where a vertex 
represents  
 an agent
and an  edge
between two vertices means that a  direct communication link exists.

In the  SDS described above, 
  sampling data of a signal $f$
 acquired by the agent $\lambda\in G$ is
 \begin{equation}\label{data.def}
y(\lambda):=\langle f, \psi_\lambda\rangle,
\end{equation}
where $\psi_{\lambda}$ is
 the impulse response of
 the agent $\lambda\in G$ (\cite{akramacha13, aldroubisiamreview01, ast2005, eldar06,  ns10, smale04,
sunsiam06, vetterli13, unser00,  vmb02}).
Fundamental signal reconstruction problems are whether and how the signal $f$ can be recovered from its sampling data $y(\lambda),\ \lambda \in G$.
For  well-posedness, the signal $f$ of interest is usually assumed to  have additional properties, such as band-limitedness, finite rate of innovation,  smoothness, and
sparse expansion in a dictionary (\cite{aldroubisiamreview01, candes06, donoho92,   donoho06,  donoho03, unser00, vmb02}). 
In this paper, we consider spatial
 signals with the following parametric representation,
 \begin{equation}\label{signalongraph.def} f:=\sum_{i\in V} c(i) \varphi_i,\end{equation}
 where   amplitudes  $c(i), i\in V$,
are  bounded,
 and   generators   $\varphi_i, i\in V$,
  are essentially supported in a  spatial neighborhood
of the  innovative position $i$. The above family of  spatial signals
appears  in magnetic resonance spectrum, mass spectrometry,  global positioning system,
  cellular radio, ultra wide-band communication,
 electrocardiogram,   and many engineering applications, see \cite{donoho06, sunaicm08, vmb02} and references therein.

 In this paper, we associate every innovative position $i\in V$ 
with some anchor agents $\lambda\in G$,
and denote  the set of such associations
 $(i, \lambda)$
 by $T$. These associations can be easily understood as 
agents within certain (spatial) range of every innovative position.
  With the above associations, we describe our  distributed  sampling and reconstruction system  (DSRS)  by an undirected (in)finite graph
\begin{equation}\label{h.graph}
{\mathcal H}:=(G\cup V, \ S\cup T\cup T^*),\end{equation}
where
$T^*=\{(\lambda, i)\in G\times V,\ (i,\lambda)\in T\}$, see Figure \ref{graphstructure.fig}.
 The above graph description of a DSRS
 plays  a crucial role for us to study signal sampling and reconstruction.

 Given a DSRS described by the above  graph ${\mathcal H}$, set
 \begin{equation}\label{signalgraph.def} E  := \{(i,i')\in V\times V,\ \  i\ne i' \ {\rm and} \
    (i, \lambda), (i',\lambda)\in T\ {\rm for \ some} \ \lambda\in G\}.
\end{equation}
We then generate a graph structure
\begin{equation}\label{v.graph} {\mathcal V}:=(V, E)\end{equation}
 for signals in \eqref{signalongraph.def}, where
 an edge between two distinct innovative positions in ${V}$  means that  a common anchor agent  exists. The above graph structure for signals 
  is different from
 the conventional one 
  in most of the literature, where the graph 
  is usually preassigned.
  The reader may refer to \cite{pesenson08,  sandryhaila2013, shuman2013} and Remark \ref{signal.remark}.
  \begin{figure}[h]
\begin{center}
\includegraphics[width=88mm, height=58mm]{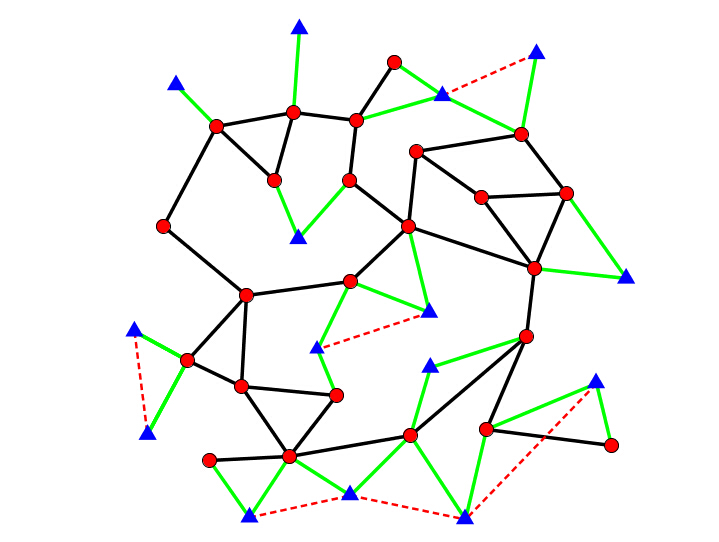}
\caption{The graph ${\mathcal H}=(G\cup V, \ S\cup T\cup T^*)$ in \eqref{h.graph} to describe a DSRS,
 where  vertices in  $G$ and  $V$ are plotted in red circles and blue triangles,
 and  edges in  $S, T$ and $E$ are in  black solid  lines, green solid lines and red dashed lines respectively.
 }
\label{graphstructure.fig}
\end{center}
\end{figure}

Define sensing matrix ${\bf S}$
 of our DSRS by
  \begin{equation}\label{sensingmatrix.def}
 {\bf S}:= (\langle \varphi_i, \psi_\lambda\rangle)_{\lambda\in G, i\in V}.
 \end{equation}
 The sensing matrix ${\bf S}$
 is stored by agents in a distributed manner. Due to the  storage limitation, each agent in our SDS  stores its corresponding  row (and perhaps also  its neighboring rows) in the sensing matrix ${\bf S}$, but it does not have the whole matrix available.
   Agents
in our SDS have limited acquisition ability  and  they could essentially catch signals not far from their
physical locations.
So the sensing matrix ${\bf S}$
 has certain   {\em polynomial off-diagonal decay}, i.e., there exist positive constants $D$ and $\alpha$ such that
\begin{equation}\label{sensingjaffard.cond}
|\langle \varphi_i, \psi_\lambda\rangle|\le D (1+\rho_{\mathcal H}(\lambda,i))^{-\alpha} \ {\rm for \ all} \  \lambda\in G \ {\rm and} \ i\in V,
\end{equation}
where  $\rho_{\mathcal H}$ is the
geodesic distance on the graph ${\mathcal H}$. 
For most  DSRSs in applications, such as
multivehicle and  multirobot networks and wireless sensor networks,
the signal generated at any innovative position could be detected by its anchor agents
and some of  their neighboring agents, but not by agents in the SDS far away.
Thus the sensing matrix ${\bf S}$ may have finite bandwidth $s\ge 0$,
\begin{equation} \label{bandwidth.intro}
 \langle \varphi_i, \psi_\lambda\rangle=0 \quad {\rm if} \ \rho_{\mathcal H}(\lambda,i)>s.
\end{equation}
 The  above global requirements \eqref{sensingjaffard.cond} and  \eqref{bandwidth.intro}  
 could be fulfilled in a distributed manner.

The sensing matrix ${\bf S}$ characterizes the sampling procedure \eqref{data.def}
of signals  with the parametric representation
\eqref{signalongraph.def}.  Applying the sensing matrix ${\bf S}$,
we obtain
the sample vector ${\bf y}=(\langle f, \psi_\lambda\rangle)_{\lambda\in G}$
of the signal $f$ 
from
its amplitude vector  ${\bf c} := (c(i))_{i\in V}$,
\begin{equation}\label{signalsamples}
{\bf y}= {\bf S} {\bf c}.
\end{equation} Under the  assumptions \eqref{sensingjaffard.cond} and  \eqref{bandwidth.intro},
it is shown in Proposition \ref{jaffard.pr} that
a signal $f$ with bounded
amplitude vector  ${\bf c}$ generates a bounded sample vector ${\bf y}$. Thus there exists a
positive constant $C$ such that
$$ \|{\bf y}\|_\infty\le C \|{\bf c}\|_\infty \ {\rm for \ all} \ {\bf c}\in \ell^\infty,$$
 where for $1\le p\le \infty$, $\ell^p$ is the space of all $p$-summable sequences  with   norm 
 $\|\cdot\|_p$.

\smallskip

A fundamental problem in sampling theory
is  the robustness of  signal reconstruction  in the presence of sampling noises
(\cite{
 bns09,  eldar06, mv05,  mvb06,  ns10, pawlak03, smale04}).
In this paper, we consider the scenario that
 the sampling data  ${\bf y}={\bf S}{\bf c}$ 
is corrupted by  bounded  deterministic/random noise $\pmb \eta=(\eta(\lambda))_{\lambda\in G}$,
\begin{equation}\label{noisysamples}
{\bf z} ={\mathbf S}{\bf c}+{\pmb \eta}
\end{equation}
(\cite{sunaicm14, wang09}).
For the robustness of our DSRS, one desires that
the signal reconstructed by some (non)linear algorithm  $\Delta$ is a suboptimal approximation
 to the original signal, in the sense that the differences between their corresponding amplitude vectors
 $\Delta({\bf z})$ and ${\bf c}$ are bounded by a multiple of noise level $\delta=\|\pmb \eta\|_\infty$, i.e.,
\begin{equation}\label{rougherror}
\|\Delta({\bf z})-{\bf c}\|_\infty\le  C \delta
\end{equation}
for some absolute  constant $C$   (\cite{adcock2013, aldroubisiamreview01,  cjs15}).

Given the noisy sampling vector ${\bf z}$ in \eqref{noisysamples}, solve the following nonlinear problem of maximal sampling error
(\cite{cadzow73, cadzow74}),
\begin{equation}\label{minimizationproblem}
 \Delta_\infty({\bf z}):=\argmin_{{\bf d}\in \ell^\infty} \|{\bf S} {\bf d}-{\bf z}\|_\infty.
\end{equation}
Observe from \eqref{noisysamples} and \eqref{minimizationproblem} that
\begin{equation*} \|{\bf S}\Delta_\infty({\bf z})-{\bf S}{\bf c}\|_\infty
  \le  
 \|{\bf S}\Delta_\infty({\bf z})-{\bf z}\|_\infty + \|\pmb\eta\|_\infty \le   \|{\bf S}{\bf c}-{\bf z}\|_\infty+ \|\pmb\eta\|_\infty\le 2\|\pmb\eta\|_\infty.\end{equation*}
Thus the solution of the $\ell^\infty$-minimization problem
   \eqref{minimizationproblem} gives a suboptimal approximation to the true amplitude vector ${\bf c}$  if
 the sensing matrix ${\bf S}$ of the DSRS has
$\ell^\infty$-stability (\cite{aldroubisiamreview01, sunxian14, unser00}).

\begin{Df}
For $1\le p\le \infty$,
 a  matrix ${\mathbf A}$ is said to have $\ell^p$-stability if there exist positive constants
$A$ and $B$ such that
\begin{equation}\label{stability.def}
A\|{\mathbf c}\|_p \le \|{\mathbf A}{\mathbf c}\|_p \le B\|{\mathbf c}\|_p \ \ {\rm for \ all}  \ {\mathbf c}\in \ell^p.
\end{equation}
We call the minimal constant $B$ and the  maximal constant $A$ for
\eqref{stability.def} to hold the upper  and lower
$\ell^p$-stability bounds respectively.
\end{Df}

The $\ell^\infty$-stability of a matrix can not be verified in a distributed manner, up to our knowledge.
We circumvent such a verification problem by showing in Theorem \ref{lpstability.tm}
that a  matrix with some polynomial off-diagonal decay has
$\ell^\infty$-stability
if it has $\ell^2$-stability.

Next we consider the problem how to  verify  $\ell^2$-stability of the sensing matrix ${\bf S}$
of our DSRS in a distributed manner.
It is well known that a finite-dimensional matrix ${\bf S}$
 has  $\ell^2$-stability if and only if ${\bf S}^T{\bf S}$ is strictly positive, and
 its  upper and lower stability bounds are  the same as  square roots of largest and  smallest eigenvalues
of ${\bf S}^T{\bf S}$.
  The above procedure
to establish  $\ell^2$-stability for the sensing matrix of our DSRS
is not feasible, because
the whole sensing matrix ${\bf S}$  is not available for any  agent in the DSRS and
there is no centralized processor to evaluate eigenvalues
of ${\bf S}^T{\bf S}$ of large size.
In Theorems \ref{l2stablility.thm1} and \ref{l2stablility.thm},
 we introduce a method to split the DSRS into a family of overlapping subsystems of small size, 
 and we show that the sensing matrix ${\bf S}$ with polynomial off-diagonal decay has $\ell^2$-stability
if and only if its quasi-restrictions
 to those  subsystems have uniform $\ell^2$-stability.
The new local criterion in  Theorems \ref{l2stablility.thm1} and \ref{l2stablility.thm}
provides a reliable tool for the  verification of the $\ell^2$-stability in a distributed manner.
 Also it  is pivotal for the design of
 a robust DSRS against
 supplement, replacement and impairment of agents,
 as it suffices to verify  the uniform stability of affected subsystems.

Then we consider signal reconstructions in a distributed manner, under the assumption that the sensing matrix $\mathbf S$ of our DSRS has $\ell^2$-stability. For centralized  signal reconstruction systems, there are many robust algorithms,
 such as the frame algorithm  and the  approximation-projection algorithm,
 to approximate  signals from their (non)linear noisy sampling data
 (\cite{Akram98, cjs15,  Christensen05,   dem08,
 Feichtinger95,
 ns10, sunsiam06, sunaicm14}).
In this paper, we develop a distributed algorithm  to  find the suboptimal approximation
  \begin{equation}\label{d2.eqn}
 \Delta_2({\bf z}):=({\bf S}^T{\bf S})^{-1} {\bf S}^T {\bf z}
 \end{equation}
to the original signal $f$ in \eqref{signalongraph.def}. 
For the case that our DSRS has finitely many agents (which is the case in most of practical applications),
the suboptimal approximation $\Delta_2({\bf z})$ in \eqref{d2.eqn}
is the unique least squares solution,
\begin{equation}\label{squareminimizationproblem}
  \Delta_2({\bf z})=\argmin_{{\bf d}\in \ell^2} \|{\bf S} {\bf d}-{\bf z}\|_2^2
  =\argmin_{{\bf d}\in \ell^2} \sum_{\lambda\in G} f_\lambda({\bf d}, {\bf z}),
\end{equation}
where  ${\bf d}=(d(i))_{i\in V}$, ${\bf z}=(z(\lambda))_{\lambda\in G}$, and
\begin{equation}
f_\lambda({\bf d}, {\bf z})=
  \Big|\sum_{i\in V} \langle \varphi_i, \psi_\lambda\rangle
 d(i)- z(\lambda)\Big|^2, \ \lambda\in G.
\end{equation}
As our SDS has strict constraints in its data processing power and communication bandwidth,  
we need 
develop distributed algorithms
    to  solve the  
    optimization problem
\begin{equation} \label{squareminimizationproblem2}
\min \sum_{\lambda\in G} f_\lambda({\bf d}, {\bf z}).
\end{equation}
For  the case that $G = V$ and the sensing matrix ${\bf S}$ is strictly diagonally dominant, the  Jacobi iterative method,
$$\left\{\begin{array}{l} d_1(\lambda)=0\\
\begin{array}{rcl} \hskip-0.1in
d_{n+1}(\lambda)
 & \hskip-0.08in = & \hskip-0.08in (\langle\varphi_\lambda, \psi_\lambda\rangle)^{-1}
\big(\sum_{i\ne \lambda} \langle \varphi_i, \psi_\lambda\rangle d_n(i)-z(\lambda)\big)\\
&\hskip-0.08in=& \hskip-0.08in {\rm argmin}_{t\in \RR} f_\lambda({\bf d}_{n; t, \lambda}, {\bf z}), \lambda\in G=V,\ n\ge 1,
\end{array}
\end{array}
\right.$$
is a distributed algorithm to solve the minimization problem \eqref{squareminimizationproblem2},
where ${\bf d}_{n; t, \lambda}$ is obtained from  ${\bf d}_n=(d_n(i))_{i\in V}$ by replacing  its $\lambda$-component
$d_n(\lambda)$  with $t$.
 The reader may refer to \cite{bertsekasbook1989,  ckl08,
 koshal2011, lopes2007,  nedic2015}
and references therein for
  historical remarks, motivations, applications and recent
advances on distributed algorithms, especially for the case that $G=V$.

In our DSRS, the set $G$ of agents  is not necessarily
 the same as the set $V$ of innovative positions,
 and even for the case that the sets $G$ and $V$ are the same,
 the sensing matrix ${\bf S}$  
 need not be strictly diagonally dominant in general.
 In this paper,  we introduce a  distributed algorithm \eqref{wn.def1} and \eqref{wn.def2}
 to  approximate $\Delta_2({\bf z})$  in
\eqref{d2.eqn}, when the sensing matrix ${\mathbf S}$ has $\ell^2$-stability  and
 satisfies the  requirements \eqref{sensingmatrix.def} and \eqref{sensingjaffard.cond}.
In the above distributed algorithm for  signal reconstruction, 
 each agent in the SDS   collects noisy observations of neighboring agents, then interacts
with its neighbors per
iteration,  and continues
the above recursive procedure until arriving at  an accurate
 approximation 
 to the solution $\Delta_2(z)$ in \eqref{d2.eqn}.
More importantly, we show in Theorems \ref{convergence.prop} and \ref{convergence.thm}
 that the proposed distributed algorithm  \eqref{wn.def1} and \eqref{wn.def2}  converges
  exponentially to the solution $\Delta_2(z)$ in \eqref{d2.eqn}. 
The establishment for the above convergence is virtually based on Wiener's lemma for localized matrices (\cite{Grochenig10, GrochenigL06, jaffard,suncasp05,suntams07,sunca11}) and on the observation
 that our sensing matrices are quasi-diagonal block dominated. 

The paper is organized as follows. In Section
\ref{sds.section}, we make some basic assumptions on the SDS and we introduce
 its Beurling dimension and sampling density. In Section
 \ref{signal.section}, we impose some 
 constraints on the graph
 ${\mathcal H}$  to describe our DSRS and then we define
  dimension and maximal rate of innovation for signals
on the graph ${\mathcal V}$. We show in Theorem \ref{Vdim.cr}
that the dimension  for signals is the
same as the Beurling dimension for the  SDS, and the
maximal rate of innovation is approximately proportional
to the sampling density  of the SDS. In Section \ref{sensing.section},
we prove in Proposition \ref{jaffard.pr} that sampling  a signal with bounded
amplitude vector by the procedure \eqref{data.def} produces a bounded sampling data vector
when the sensing matrix of the SDS has certain
polynomial off-diagonal decay. In Section \ref{stablesampling.section}, we establish
in Theorem \ref{lpstability.tm} that
if a matrix with certain off-diagonal decay has $\ell^2$-stability then it has $\ell^p$-stability for all $1\le p\le \infty$, and also in Theorem \ref{leastsqaure.tm}
that the solution $\Delta_2({\bf z})$ in \eqref{d2.eqn} is a suboptimal approximation to the
true amplitude vector. In   Theorems
\ref{l2stablility.thm1} and \ref{l2stablility.thm} of Section \ref{criterion.section},
we introduce a  criterion for the $\ell^2$-stability of a sensing matrix, that could be verified in a distributed manner.
In Section \ref{distributedalgorithm.section}, we propose a  distributed algorithm
to solve
the minimization problem \eqref{squareminimizationproblem}.
In Section \ref{simulations.section},
we present simulations to demonstrate our proposed algorithm  for robust signal reconstruction.
In Section \ref{proofs.section}, we include  proofs of all conclusions.

The sampling theory developed in this paper enjoys
the advantages of scalability  of  network sizes and data privacy preservation. Some results of this paper were announced in \cite{cjssampta}.

Notation: ${\bf A}^T$ is the transpose of a matrix ${\mathbf A}$;
$\|c\|_p$ is the norm on $\ell^p$; $\chi_F$ is the index function on a set $F$;
$\lceil x\rceil$ is the ceiling of $x \in \RR$; $\lfloor x\rfloor$ is the floor of $x \in \RR$;
$\# F$ is the cardinality of  a set $F$; and  $\|{\mathbf A}\|_{{\mathcal B}^2}$  is the operator norm of a matrix ${\bf A}$ on $\ell^2$.

\section{Spatially distributed 
systems}
\label{sds.section}

Let ${\mathcal G}$ be  the graph in \eqref{g.graph} 
 to  describe our  SDS. In this paper, we always assume that
 ${\mathcal G}$ is {\em connected} and {\em simple} (i.e., undirected, unweighted,  no graph loops nor multiple edges),
  which can be interpreted as follows:
   \begin{itemize}
   \item Agents in the SDS can communicate across the entire
network, but they have direct communication links only to adjacent agents. 
  \item Direct communication links between agents are bidirectional.
      \item
Agents have the same communication specification.
\item The communication  component 
    is not used for data  transmission within an agent.
\item
No multiple direct communication channels between agents exists.
   \end{itemize} 

In this section,
we recall  geodesic distance 
 on the graph ${\mathcal G}$
to measure 
communication cost between agents. Then
we consider doubling and polynomial growth properties of the counting measure 
on the graph ${\mathcal G}$,  and  we
 introduce Beurling dimension  and sampling density of the SDS. 
 For a discrete sampling set in the $d$-dimensional Euclidean space, 
  the reader may refer to \cite{cks08, dhsw11} for its Beurling dimension and to \cite{aldroubisiamreview01, ns10, sunsiam06,  unser00} for its sampling  density.
Finally, 
we introduce a special family of balls to cover the graph $\mathcal G$,
which will be used in Section
\ref{distributedalgorithm.section} for
the consensus of our proposed distributed algorithm.

\subsection{Geodesic distance and communication cost}
\label{distance.subsection}

 For a connected simple graph ${\mathcal G}:=(G,S)$,  let $\rho_{\small {\mathcal G}}(\lambda,\lambda)=0$ for $\lambda\in G$, and $\rho_{\small {\mathcal G}}(\lambda,\lambda')$ be the  number of edges
 in a shortest path  connecting two distinct vertices $\lambda,\lambda'\in G$.
The above function
 $\rho_{\small {\mathcal G}}$ on $G\times G$   is known as
   {\em geodesic distance} on the graph ${\mathcal G}$ (\cite{chungbook}).
It is  nonnegative and symmetric:
\begin{itemize}
\item [{(i)}] 
 $\rho_{\small {\mathcal G}} (\lambda,\lambda') \ge 0$ for all $\lambda,\lambda' \in G$;

 \item[{(ii)}] 
$\rho_{\small {\mathcal G}}(\lambda,\lambda')=\rho_{\small {\mathcal G}}(\lambda',\lambda)$ for all $\lambda,\lambda'\in G$.
\end{itemize}
And it satisfies
identity of indiscernibles and the triangle inequality:
\begin{itemize}

\item [{(iii)}] 
 $\rho_{\small {\mathcal G}} (\lambda,\lambda')=0$ if and only if $\lambda=\lambda'$;

\item[{(iv)}] 
 $\rho_{\small {\mathcal G}}(\lambda,\lambda')\le \rho_{\small {\mathcal G}}(\lambda,\lambda^{\prime\prime})+\rho_{\small {\mathcal G}}(\lambda^{\prime\prime},\lambda')$ for all $\lambda,\lambda',\lambda^{\prime\prime} \in G$.

\end{itemize}

Given two nonadjacent agents $\lambda$ and $ \lambda'\in G$,
  the distance $\rho_{\mathcal G}(\lambda, \lambda')$ can be used to measure 
   the communication cost between these two agents  if the communication is processed through  their shortest path.

\subsection{Counting measure, Beurling dimension and sampling density}
\label{counting.subsection}

For a connected simple
 graph ${\mathcal G}:=(G,S)$,
 denote  its {\em counting measure} by $\mu_{\mathcal G}$,
$$\mu_{\mathcal G} (F) := \sharp(F) \ \ {\rm for}  \ F \subset  G.$$

\begin{Df}\label{doublingmeausre.def}
 The counting measure $\mu_{\mathcal G}$  is said to be a doubling measure if there exists a positive number $D_0({\mathcal G})$ such that
\begin{equation}\label{doublingconstants.def}
 \mu_{\mathcal G}(B_{\mathcal G}(\lambda,2r)) \le  D_0({\mathcal G})\mu_{\mathcal G}(B_{\mathcal G}(\lambda,r)) \  {\rm \ for \  all } \   \lambda\in G \ {\rm and} \ r \ge 0,\ \end{equation}
 where
 $$B_{\mathcal G}(\lambda,r):=\{\lambda'\in G,\ \ \rho_{\small {\mathcal G}}(\lambda,\lambda') \le r\}$$
 is the closed ball with center $\lambda$ and radius $r$.
 \end{Df}

The doubling property of the counting measure $\mu_{\mathcal G}$ can be
interpreted as numbers of agents in $r$-neighborhood and $(2r)$-neighborhood of any agent are comparable.
 The  doubling constant of 
 $\mu_{\mathcal G}$ is the minimal
  constant $D_0({\mathcal G})\ge 1$ for \eqref{doublingconstants.def} to hold (\cite{coifman71, dh08}). 
It
dominates  the maximal vertex degree 
of the graph ${\mathcal G}$,
\begin{equation}\label{defd0.eqn}
\deg({\mathcal G})\le D_0({\mathcal G}),\end{equation}
because
\begin{equation*} \deg({\mathcal G})
= 
 \max_{\lambda\in G}
\#\{\lambda'\in G,\ (\lambda,\lambda')\in S\} 
 \le \max_{\lambda\in G} \# (B_{\mathcal G}(\lambda,1)\big)\le D_0({\mathcal G}).
\end{equation*}
We remark that for a finite graph ${\mathcal G}$, its doubling constant $D_0({\mathcal G})$
could be much larger than
 its maximal vertex degree $\deg({\mathcal G})$. For instance,
a tree  with one branch for the first $L$ levels
 and two branches for the next $L$ levels
 has $3$ as its maximal vertex degree and $(2^{L+1}+L-1)/(L+1)$ as its doubling constant, see Figure \ref{tree.fig} with $L=3$.
 \begin{figure}[h]
\begin{center}
\includegraphics[width=78mm, height=38mm]{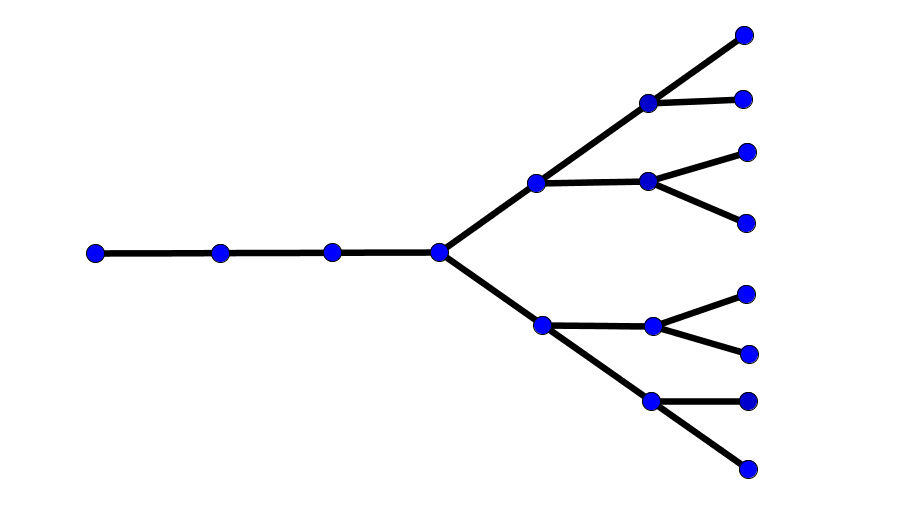}
\caption{A tree with large doubling constant but limited maximal vertex degree.
 }
\label{tree.fig}
\end{center}
\end{figure}

The counting measure on  an infinite graph is not necessarily a doubling  measure.
However,  the counting measure on a finite graph is a doubling  measure
 and its  doubling constant  could depend on the local topology and size of the graph, cf. the tree in
Figure \ref{tree.fig}.
In this paper, the graph ${\mathcal G}$ to describe our SDS
is assumed to have  its  counting measure
 with the doubling property \eqref{doublingconstants.def}.

\noindent{\bf Assumption 1}: {\em The counting measure $\mu_{\mathcal G}$ of the graph ${\mathcal G}$ is a doubling measure,
\begin{equation}\label{basicassumption0}
D_0({\mathcal G})<\infty.
\end{equation}}
 Therefore the maximal vertex degree of graph $\mathcal G$ 
 is finite, 
 \begin{equation*} 
 \deg({\mathcal G})<\infty,
 \end{equation*}
 which could be understood as that there are limited direct communication channels for every  agent in the SDS.

\begin{Df}\label{polynomial.def}
 The counting measure $\mu_{\mathcal G}$  is said to have polynomial growth if
there exist positive constants $D_1({\mathcal G})$ and $d({\mathcal G})$ such that
   \begin{equation}\label{countmeasure.pr.eq1}
   \mu_{\mathcal G}(B_{\mathcal G}(\lambda,r))\le D_1({\mathcal G})(1+r)^{d({\mathcal G})} \  \ {\rm  for \ all} \ \lambda \in G\ {\rm and} \ r\ge 0.\end{equation}
\end{Df}
For the graph ${\mathcal G}$ associated with an SDS,
we may consider minimal constants $d({\mathcal G})$
and $D_1({\mathcal G})$ in \eqref{countmeasure.pr.eq1}
as  {\em Beurling dimension}
 and {\em sampling density} of the SDS 
 respectively. We remark that
 \begin{equation}\label{dimension.remark}
 d({\mathcal G})\ge 1,
 \end{equation}
 because
$$\sup_{\lambda\in G} \mu_{\mathcal G}(B_{\mathcal G}(\lambda,r))\ge 1+r \ \ {\rm for \ all} \  0\le r\le {\rm diam}({\mathcal G}),$$
where ${\rm diam}({\mathcal G}):=\sup_{\lambda,\lambda'\in G} \rho_{\mathcal G}(\lambda,\lambda')$ is the diameter of the graph ${\mathcal G}$.

 \smallskip
 Applying  \eqref{doublingconstants.def} repeatedly leads to the following general doubling property:
 \begin{equation*}
 \mu_{\mathcal G}(B_{\mathcal G}(\lambda,sr))
 \le (D_0({\mathcal G}))^{\lceil \log_2 s\rceil} \mu_{\mathcal G}(B_{\mathcal G}(\lambda,r))
  \le D_0({\mathcal G}) s^{\log_2 D_0({\mathcal G})} \mu_{\mathcal G}(B_{\mathcal G}(\lambda,r))
 \end{equation*}
   for all $\lambda\in G$, $s\ge 1$ and $r\ge 0$.
Thus
\begin{equation*} 
\mu_{\mathcal G}(B_{\mathcal G}(\lambda, r))
\le D_0({\mathcal G}) (1+r)^{\log_2 D_0({\mathcal G})}  \mu_{\mathcal G}\Big(B_{\mathcal G}\Big(\lambda, \frac{r}{1+r}\Big)\Big) 
=D_0({\mathcal G}) (1+r)^{\log_2 D_0({\mathcal G})}, \ r\ge 0.\end{equation*}
This shows that a doubling measure 
 has polynomial growth.

  \begin{pr} \label{countmeasure.pr}
  If the counting measure $\mu_{\mathcal G}$ on a
   connected simple graph ${\mathcal G}$ is a doubling measure, then it has polynomial growth.
  \end{pr}

For a connected simple graph ${\mathcal G}$, its maximal vertex degree 
is finite if the counting measure $\mu_{\mathcal G}$ has polynomial growth, 
but the converse is  not true. 
 We observe that if the maximal vertex degree $\deg({\mathcal G})$ is finite, then
 the counting measure $\mu_{\mathcal G}$ has exponential growth,
\begin{equation}\label{maximalvertex.def} \mu_{\mathcal G}(B_{\mathcal G}(\lambda,r))\le \frac{(\deg({\mathcal G}))^{r+1}-1}{\deg({\mathcal G})-1}
 \  \ {\rm  for \ all} \ \lambda \in G\ {\rm and} \ r\ge 0.
\end{equation}


\subsection{Spatially distributed subsystems}
\label{covering.subsection}

 For a connected simple graph ${\mathcal G}:=(G,S)$, take  a maximal $N$-disjoint subset $G_N \subset G, 0 \le N\in \RR$,
such that
 \begin{equation}\label{vn.def1}
 B_{\mathcal G}(\lambda, N)\medcap  \big(\medcup_{\lambda_m\in G_N}B_{\mathcal G}(\lambda_m, N)\big) \neq \emptyset \  \ {\rm for\  all}\  \lambda\in G,
  \end{equation}
and
   \begin{equation}\label{vn.def2}
 B_{\mathcal G}(\lambda_m, N)\medcap B_{\mathcal G}(\lambda_{m'}, N)= \emptyset \ \ {\rm for \ all}\  \lambda_m, \lambda_{m'}\in G_N.\end{equation}
 For $0 \le N <1$, it follows from \eqref{vn.def1} that $G_N=G$. For $N\ge 1$,
 there are many subsets $G_N$ of vertices satisfying \eqref{vn.def1} and \eqref{vn.def2}.
For instance, we can construct $G_N=\{\lambda_m\}_{m\ge 1}$ as follows:  take a $\lambda_1\in G$
and define $\lambda_m,\ m\ge 2$, recursively by 
\begin{equation*}
   \lambda_m = \mathop{\argmin}_{\lambda \in A_m}{\rho_{\mathcal G}(\lambda, \lambda_1)},
   \end{equation*}
where $A_m=\{\lambda\in G, \ B_{\mathcal G}(\lambda, N)\cap B_{\mathcal G}(\lambda_{m'}, N)=\emptyset, 1\le m'\le m-1\}$.

For a set $G_N$ satisfying \eqref{vn.def1} and \eqref{vn.def2},
the family of  balls $\{B_{\mathcal G}(\lambda_m, N'), \lambda_m\in G_N\}$   with $N'\ge 2N$
  provides a finite covering for  $G$.

 \begin{pr}\label{covering.pr}
  Let  ${\mathcal G}:=(G, S)$
   be a connected simple graph and  $\mu_{\mathcal G}$ have the doubling property \eqref{basicassumption0} with constant $D_0({\mathcal G})$.
 If $G_N$ satisfies
\eqref{vn.def1} and \eqref{vn.def2},
then 
\begin{equation}\label{covering.pr.eq1} 1
\le \inf_{\lambda\in G}\Sigma_{\lambda_m \in G_N}\chi_{B_{\mathcal G}(\lambda_m, N')}(\lambda) 
\le \sup_{\lambda\in G}\Sigma_{\lambda_m \in G_N}\chi_{B_{\mathcal G}(\lambda_m, N')}(\lambda)
\le (D_0({\mathcal G}))^{\lceil\log_2 (2N'/N+1)\rceil}
\end{equation}
for all $N'\ge 2N$.
 \end{pr}
 For $N'\ge 0$, define a family of
spatially distributed subsystems
 $${\mathcal G}_{\lambda, N'}:= (
 B_{\mathcal G}(\lambda, N'), S_{\lambda, N'})$$
 with fusion agents $\lambda\in G_N$,
 where $(\lambda', \lambda^{\prime\prime})\in S_{\lambda, N'}$
 if $\lambda', \lambda^{\prime\prime}\in B_{\mathcal G}(\lambda, N')$ and $(\lambda', \lambda^{\prime\prime})\in S$.
Then the maximal $N$-disjoint property of the set $G_N$ means that the $N$-neighboring
 subsystems ${\mathcal G}_{\lambda_m, N}, \lambda_m\in G_N$, have no common agent.
 On the other hand, it follows from Proposition \ref{covering.pr} that for any $N'\ge 2N$,
 every agent in our SDS is in at least one
  and at most finitely many of the $N'$-neighboring
 subsystems ${\mathcal G}_{\lambda_m, N'}, \lambda_m\in G_N$.
 The above idea to  split the SDS  into
 subsystems of small sizes is crucial in  our proposed distributed algorithm  in  Section \ref{distributedalgorithm.section} for stable signal reconstruction. 

\section{Signals on the graph $\mathcal V$} 
\label{signal.section}

Let $V$ be the set
of innovative positions of signals
$f$ in \eqref{signalongraph.def}, and ${\mathcal G = (G, S)}$ be the graph in \eqref{g.graph} to represent our SDS. We build
the graph ${\mathcal H}$ 
in  \eqref{h.graph} to describe our DSRS
by associating every innovative position in $V$  with some anchor agents  in $G$.
In this paper, we consider DSRS with the following properties.

\noindent {\bf Assumption 2}: {\em There is a direct communication link between distinct anchor agents of an innovative position,
\begin{equation}\label{basicassumption1}
(\lambda_1, \lambda_2)\in S \ {\rm if} \ (i, \lambda_1) \ {\rm and}\ (i, \lambda_2)\in T \ {\rm for \ some} \ i\in V.
\end{equation}}

 \noindent {\bf Assumption 3}: {\em There are finitely many innovative positions for any anchor agent,
\begin{equation}\label{basicassumption2}
L:=\sup_{\lambda\in G} \#\{i\in V, \ (i, \lambda)\in T\}<\infty.
\end{equation}}

 \noindent {\bf Assumption 4}: {\em Any agent has an anchor agent within bounded distance,}
\begin{equation}\label{basicassumption3}
M:= \sup_{\lambda\in G} \inf \{\rho_{\mathcal G}(\lambda, \lambda'),  \ (i, \lambda')\in T \ {\rm for \ some}\ i\in V\}<\infty.
\end{equation}

The  graph ${\mathcal H}$ associated with the above DSRS is a connected simple graph. Moreover, we have the following important properties about
shortest paths between different vertices in ${\mathcal H}$.

\begin{pr}\label{shortestpath.pr}
Let  the graph ${\mathcal H}$ in \eqref{h.graph} satisfy \eqref{basicassumption1}. Then all intermediate vertices  in the shortest paths in ${\mathcal H}$ to connect distinct vertices in ${\mathcal H}$
belong to the subgraph ${\mathcal G}$.
\end{pr}

By Proposition \ref{shortestpath.pr},
\begin{equation}\label{hgball.eq}
\rho_{\mathcal H}(\lambda, \lambda') = \rho_{\mathcal G}(\lambda, \lambda') \quad {\rm for \ all}\ \lambda, \lambda'\in G,
\end{equation}
and
\begin{equation}\label{rhorhoH.eq}
\rho_{\mathcal H}(i,i')=2+\inf_{\lambda, \lambda'\in G} \{\rho_{\mathcal G}(\lambda, \lambda'):
(i,\lambda), (i', \lambda')\in T\} \quad {\rm for \ all\ distinct} \  i, i'\in V,
\end{equation}
 where  $\rho_{\mathcal H}$ is the  geodesic distance
for the
graph
${\mathcal H}$.

Let
${\mathcal V}$ be the graph in \eqref{v.graph},
 where there is an edge between
 two distinct innovative positions if they share a common anchor agent.
One may easily verify that the graph ${\mathcal V}$  is undirected  
and its maximal vertex degree 
 is finite,
\begin{equation}\label{degV.eq}
\deg ({\mathcal V}) \le 
L  \sup_{i\in V} \#\{\lambda\in G, \ (i,\lambda)\in T\} 
\le  L (\deg ({\mathcal G})+1)
\end{equation}
by \eqref{defd0.eqn}, \eqref{basicassumption0}, \eqref{basicassumption1} and \eqref{basicassumption2}.


We cannot define a  geodesic distance 
on ${\mathcal V}$
as in Subsection \ref{distance.subsection}, since the graph ${\mathcal V}$ is unconnected in general.
With the help of the graph ${\mathcal H}$ to describe our DSRS, we define a distance $\rho$ on the graph ${\mathcal V}$.
%

\begin{pr}\label{Vdistance.pr}
Let ${\mathcal H}$ be the graph in \eqref{h.graph}. Define a function $\rho: V\times V \longmapsto \RR$ by
\begin{equation}\label{rho.def}
\rho(i,i')=\left\{\begin{array}{ll}
0 & {\rm if }\  i = i'\\
\rho_{\mathcal H}(i,i')-1 &  {\rm if} \ i\neq i'.\end{array}\right.
\end{equation}
If the graph $\mathcal H$ satisfies \eqref{basicassumption1}, then
$\rho$ is  a distance
on the graph ${\mathcal V}$:
\begin{itemize}
\item [{(i)}] 
 $\rho (i,i') \ge 0$ for all $i,i'\in V$;

 \item[{(ii)}] 
$\rho(i,i')=\rho(i',i)$ for all $i,i'\in V$;

\item [{(iii)}] 
 $\rho(i,i')=0$ if and only if $i=i'$; and

\item[{(iv)}] 
 $\rho(i,i')\le \rho(i,i^{\prime\prime})+\rho(i^{\prime\prime},i')$ for all $i,i',i^{\prime\prime} \in V$.

\end{itemize}
\end{pr}
Clearly, the above distance between two endpoints of an edge in $\mathcal V$ is one.
Denote  the closed ball with center $i\in V$ and radius $r$ by
$$B(i,r)=\{i'\in V,\ \rho(i,i')\le r\},$$
and  the counting measure on $V$ by $\mu$.
We say that $\mu$ is a {\em doubling measure} if 
\begin{equation}\label{mudoubling.def}
\mu(B(i, 2r))\le D_0 \mu(B(i,r)) \ {\rm for \ all}\ i\in V \ {\rm and} \ r\ge 0,
\end{equation}
and it has {\em polynomial growth}
if
\begin{equation}\label{mupolynomial.def}
\mu(B(i,r))\le D_1 (1+r)^{d}\ {\rm for \ all}\ i\in V \ {\rm and} \ r\ge 0,
\end{equation}
where $D_0$, $D_1$ and $d$ are positive constants.
The minimal constant $D_0$ for \eqref{mudoubling.def} to hold is known as the doubling constant,
and the minimal constants $d$ and $D_1$ in \eqref{mupolynomial.def}
are called {\em dimension} and {\em maximal rate of innovation} for signals on the graph ${\mathcal V}$ respectively. The concept of rate of innovation  was introduced
in \cite{vmb02} and later extended in \cite{sunaicm08, sunxian14}.
The reader may refer to \cite{bns09, bdvmc08, dvb07, mv05, pbd14,  sd07, sunsiam06, sunaicm08,    sunxian14,   vmb02} and references therein  for sampling and reconstruction of  signals with  finite rate of innovation.
\smallskip

In the next two propositions, we show that  the counting measure $\mu$ on ${\mathcal V}$
has the doubling property (respectively, the  polynomial growth property) if and only if
  the counting measure $\mu_{\mathcal G}$  on ${\mathcal G}$ does.

\begin{pr}\label{Vdoublingmeasure.pr}
Let  ${\mathcal G}$ and ${\mathcal H}$
satisfy Assumptions 1 -- 4.
If $\mu_{\mathcal G}$ is a doubling measure with constant $D_0({\mathcal G})$, then
\begin{equation}\label{Vdoublingmeasure.pr.eq1}
\mu(B(i, 2r))\le L (D_0({\mathcal G}))^2
\Big(\frac{(\deg ({\mathcal G}))^{2M+3}-1}{\deg({\mathcal G})-1}\Big)
 \mu(B(i,r)) \  \ {\rm  for\  all} \  i\in V\  {\rm  and} \ r\ge 0.
\end{equation}
 Conversely, if
$\mu$ is a doubling measure with constant $D_0$,
then
\begin{equation}\label{Vdoublingmeasure.pr.eq2}
\mu_{\mathcal G}(B_{\mathcal G}(\lambda, 2r))\le LD_0^2
\Big(\frac{(\deg({\mathcal G}))^{2M+3}-1}{\deg({\mathcal G})-1}\Big)^2 \mu_{\mathcal G}(B_{\mathcal G}(\lambda,r))
\  \ {\rm for\ all} \  \lambda\in G \ {\rm  and} \ r\ge 0.\end{equation}
\end{pr}

\begin{pr}
 \label{Vpolynomial.pr}
Let  ${\mathcal G}$ and ${\mathcal H}$
satisfy Assumptions 1 -- 4. 
 If $\mu_{\mathcal G}$ has polynomial growth  with Beurling dimension $d({\mathcal G})$ and
  sampling density $D_1({\mathcal G})$, then
 \begin{equation}\label{Vpolynomial.pr.eq1}
 \mu(B(i,r))\le  L D_1({\mathcal G}) (1+r)^{d({\mathcal G})}\ \ {\rm for\ all} \ i\in V\ {\rm  and} \  r\ge 0.
 \end{equation}
 Conversely, if $\mu$ has polynomial growth with dimension $d$ and
 maximal rate of innovation $D_1$, then
 \begin{equation}\label{Vpolynomial.pr.eq2}
 \mu_{\mathcal G}(B_{\mathcal G}(\lambda, r))\le  2^d  \Big(\frac{(\deg ({\mathcal G}))^{2M+3}-1}{\deg({\mathcal G})-1}\Big) D_1 (1+r)^d  \ \ {\rm  for\ all} \ \lambda\in G\ {\rm  and} \ r\ge 0.
 \end{equation}
 \end{pr}

By \eqref{dimension.remark}, Propositions \ref{Vdoublingmeasure.pr} and \ref{Vpolynomial.pr},
 we  conclude that signals in \eqref{signalongraph.def} have their dimension $d$ being the same as the Beurling dimension $d({\mathcal G})$, and their maximal rate $D_1$ of innovation being approximately proportional to
 the sampling density  $D_1({\mathcal G})$. 

\begin{Tm}\label{Vdim.cr}
 Let  ${\mathcal G}$ and ${\mathcal H}$ satisfy Assumptions 1 -- 4.
Then
\begin{equation} \label{Vdim.cr.eq1}
d({\mathcal G})=d\ge 1
\end{equation}
and
\begin{equation}\label{Vdim.cr.eq2}
L^{-1} D_1\le D_1({\mathcal G})\le 2^d  \Big(\frac{(\deg ({\mathcal G}))^{2M+3}-1}{\deg({\mathcal G})-1}\Big)D_1.
\end{equation}
\end{Tm}

We finish this section with a remark about signals on our graph $\mathcal V$, cf.   \cite{pesenson08,  sandryhaila2013, shuman2013}.

\begin{re}\label{signal.remark} {\rm
Signals on the graph ${\mathcal V}$ are analog in nature, while
signals on graphs in most of the literature are discrete (\cite{pesenson08,  sandryhaila2013, shuman2013}).
Let  ${\bf p}_\lambda$  and ${\bf p}_i$ be the physical positions of the agent $\lambda\in G$
and  innovative position $i\in V$, respectively.
If there exist positive constants $A$ and $B$ such that
\begin{equation*}
A \sum_{i\in V} |c(i)|^2\le \sum_{i\in V}| f({\bf p}_i)|^2+ \sum_{\lambda\in G} |f({\bf p}_\lambda)|^2\le B
\sum_{i\in V} |c(i)|^2
\end{equation*}
for all signals $f$ with the parametric representation  \eqref{signalongraph.def}, then we can  establish a one-to-one correspondence between the analog signal $f$   and the discrete signal $F$ on the graph ${\mathcal H}$,
where
$$F(u)= f({\bf p}_u),\ u\in G\cup V.$$
The above family of discrete signals $F$ forms a  linear space, which could
be a  Paley-Wiener space associated with some positive-semidefinite operator (such as Laplacian) on the graph ${\mathcal H}$.
Using the above correspondence, our theory for signal sampling and reconstruction applies
by assuming that the impulse response $\psi_\lambda$ of every agent
$\lambda\in G$  is supported on ${\bf p}_u,  u \in G \cup V$.}
\end{re}

\section{Sensing matrices with polynomial off-diagonal decay}
\label{sensing.section}

Let ${\mathcal  H}$ be the
connected simple graph in \eqref{h.graph} to describe our  DSRS,
and  the  sensing matrix ${\bf S}$  associated with the DSRS be as in \eqref{sensingmatrix.def}.
As agents
in the DSRS have limited sensing ability,
 we assume in this paper that the sensing matrix ${\bf S}$ in \eqref{sensingmatrix.def}
satisfies 
\begin{equation}\label{sensingmatrix.assumption}
{\bf S}\in {\mathcal J}_{\alpha}({\mathcal G}, {\mathcal V}) \ {\rm for\ some}\ \alpha > d,
\end{equation}
where
  \begin{equation}\label{jaffard.def1}
  {\mathcal J}_{\alpha}({\mathcal G}, {\mathcal V}):=\big \{ {\bf A}:=(a(\lambda,i))_{\lambda\in G,i\in V}, \ \|{\bf A}\|_{{\mathcal J}_{\alpha}({\mathcal G}, {\mathcal V})}<\infty \big \} \end{equation}
is  the
{\em Jaffard class} ${\mathcal J}_{\alpha}({\mathcal G}, {\mathcal V})$ of  matrices with
polynomial off-diagonal decay, and
 \begin{equation}\label{jaffard.def2}
 \|{\bf A}\|_{{\mathcal J}_{\alpha}({\mathcal G}, {\mathcal V})} :=   \sup_{\lambda\in G,i\in V}(1+\rho_{\mathcal H}(\lambda,i))^{\alpha}|a(\lambda,i)|, \quad \alpha\ge 0.
  \end{equation}
  The reader may refer to \cite{    Grochenig10, GrochenigL06, jaffard,  suncasp05, suntams07, sunca11} for matrices with various off-diagonal decay.

\smallskip

We observe that a matrix in  ${\mathcal J}_\alpha({\mathcal G}, {\mathcal V}), \alpha>d$, defines a bounded operator from $\ell^p(V)$ to $\ell^p(G), 1\le p\le\infty$.

 \begin{pr}\label{jaffard.pr}
  Let  ${\mathcal G}$ and ${\mathcal H}$
satisfy Assumptions 1 -- 4,  
${\mathcal V}$ be as in \eqref{v.graph}, and
 let  $\mu_{\mathcal G}$
 have polynomial growth  
  with  Beurling dimension $d$ and  sampling density $D_1({\mathcal G})$.
  If
 ${\mathbf A}\in {\mathcal J}_\alpha({\mathcal G}, {\mathcal V})$ for some $\alpha>d$, then
\begin{equation} \label{jaffardpr.eq1}
\|{\mathbf A}{\mathbf c}\|_p\le  \frac{D_1({\mathcal G}) L\alpha}{\alpha-d}\|{\mathbf A}\|_{{\mathcal J}_\alpha({\mathcal G}, {\mathcal V})} \|{\bf c}\|_p
\ \ {\rm for\ all} \ {\mathbf c}\in \ell^p, 1\le p\le \infty.
\end{equation}
  \end{pr}

 For a DSRS with its sensing matrix in ${\mathcal J}_\alpha({\mathcal G}, {\mathcal V})$, we obtain from
 \eqref{signalsamples} and
   Proposition \ref{jaffard.pr}
    that a signal with bounded amplitude vector generates a bounded sampling data vector.

\smallskip

Define  band matrix approximations of a matrix ${\mathbf A}= (a(\lambda,i))_{\lambda\in G,i\in V}$ by 
\begin{equation}\label{bandapproximation.def}
{\mathbf A}_s:=(a_s(\lambda,i))_{\lambda\in G,i\in V},\  s\ge 0, 
\end{equation}
where
$$a_s(\lambda,i)=\left\{\begin{array}{ll} a(\lambda,i) & {\rm if}\ \ \rho_{\mathcal H}(\lambda,i)\le s\\
 0 &  {\rm if}\ \ \rho_{\mathcal H}(\lambda,i)>s.\end{array}
 \right.$$
We say a matrix $\mathbf A$ has {\em bandwidth} $s$ if $\mathbf A = \mathbf A_s$. Clearly, any matrix $\mathbf A$ with bounded entries and bandwidth $s$ belongs to Jaffard class ${\mathcal J}_{\alpha}({\mathcal G}, {\mathcal V})$,
  $$\|{\bf A}\|_{{\mathcal J}_{\alpha}({\mathcal G}, {\mathcal V})}
\le (s+1)^{\alpha} \|{\bf A}\|_{{\mathcal J}_{0}({\mathcal G}, {\mathcal V})} \ {\rm \ for \ all \ }  \alpha \ge 0.$$
 In our DSRS, the sensing matrix $\mathbf S$ has bandwidth $s$ means that any agent can only detect signals at innovative positions within their geodesic distance less than or equal to $s$.
 In the next proposition, we show that matrices in the Jaffard class
 can be well approximated by band matrices.

\begin{pr}\label{bandapproximation.pr}
 Let  graphs ${\mathcal G}$, ${\mathcal H}$,  ${\mathcal V}$, $d$ and $D_1({\mathcal G})$
 be as in Proposition \ref{jaffard.pr}.
 If
 ${\mathbf A}\in {\mathcal J}_\alpha({\mathcal G}, {\mathcal V})$ for some $\alpha>d$,
 then
 \begin{equation}\label{bandapproximation.pr.eq1}
 \|({\mathbf A}-{\mathbf A}_s) {\mathbf c}\|_p\le \frac{ D_1({\mathcal G}) L\alpha}{\alpha-d} (s+1)^{-\alpha+d} \|{\mathbf A}\|_{{\mathcal J}_\alpha({\mathcal G}, {\mathcal V})}\|{\mathbf c}\|_p
 \ \ {\rm  for\ all} \ {\bf c}\in \ell^p, 1\le p\le \infty,
\end{equation}
where
 ${\mathbf A}_s, s\ge 1$, are band matrices in \eqref{bandapproximation.def}.
\end{pr}

The above band matrix approximation property
will be used later  in the establishment of a local stability criterion in Section \ref{criterion.section}
 and exponential convergence of a distributed  reconstruction algorithm in Section \ref{distributedalgorithm.section}.

\section{Robustness of distributed
 sampling and reconstruction systems}
\label{stablesampling.section}

 Let  ${\bf S}$ be the sensing matrix associated with our DSRS. 
We say that a reconstruction algorithm $\Delta$ is a {\em perfect reconstruction}
in noiseless environment if
\begin{equation}\label{perfect.condition}
\Delta({\bf S} {\bf c})={\bf c}\ \ {\rm for \ all}\  {\bf c}\in \ell^\infty.
\end{equation}
In this section, we first study robustness of the DSRS in term of the $\ell^\infty$-stability. 

 \begin{pr}\label{linftystability.pr}
   Let  ${\mathcal G}$ and ${\mathcal H}$
satisfy Assumptions 1 -- 4,
${\mathcal V}$ be as in \eqref{v.graph},
$\mu_{\mathcal G}$
 have polynomial growth with Beurling dimension $d$, and let $\mathbf S$ satisfy \eqref{sensingmatrix.assumption}.
 Then  there is   a  reconstruction algorithm $\Delta$ with the suboptimal approximation property
 \eqref{rougherror} and
 the perfect reconstruction property \eqref{perfect.condition}
 if and only if
  ${\bf S}$ has $\ell^\infty$-stability.
 \end{pr}

 The sufficiency in Proposition \ref{linftystability.pr} holds by taking $\Delta=\Delta_\infty$ in \eqref{minimizationproblem}, while the
 necessity follows by applying \eqref{rougherror} to
   $\pmb\eta={\bf S}{\bf d}$ with ${\bf d}\in \ell^\infty$.

The $\ell^\infty$-stability of a matrix can not be verified in a distributed manner, up to our knowledge.
In the next theorem,
we circumvent such a verification problem by reducing  $\ell^\infty$-stability of
a  matrix in Jaffard class 
  to its $\ell^2$-stability,
 for which a distributed verifiable criterion will be provided  in Section \ref{criterion.section}.

\begin{Tm}\label{lpstability.tm}
 Let  ${\mathcal G}, {\mathcal H}, {\mathcal V}$ and   $d$ be as in Proposition  \ref{linftystability.pr}, and let
 ${\mathbf A}\in {\mathcal J}_\alpha({\mathcal G}, {\mathcal V})$ for some $\alpha>d$.
 If ${\mathbf A}$ has $\ell^2$-stability, then
 it has $\ell^p$-stability for all $1\le p\le \infty$.
\end{Tm}

The reader may refer to  \cite{akramjfa09, shincjfa09, sunca11}
for equivalence of  $\ell^p$-stability of localized  matrices
 for different $1\le p\le \infty$.
The lower and upper $\ell^p$-stability
 bounds of the matrix ${\bf A}$
depend on its  $\ell^2$-stability bounds and local features of the graph ${\mathcal H}$.
From the proof of Theorem \ref{lpstability.tm}, we  observe that  they
depend only on
the $\ell^2$-stability bounds,   ${\mathcal J}_\alpha({\mathcal G}, {\mathcal V})$-norm of the  matrix
${\mathbf A}$,
maximal vertex degree ${\rm deg}({\mathcal G})$, the Beurling dimension $d$, the sampling density $D_1({\mathcal G})$, and the constants $L$ and $M$ in \eqref{basicassumption2} and \eqref{basicassumption3}. So the sensing matrix of our DSRS may have its  $\ell^p$-stability bounds
 independent of the size of the DSRS.

\smallskip

For  the graph ${\mathcal V}$ in \eqref{v.graph}
 and the distance $\rho$  in \eqref{rho.def},
define 
  \begin{equation}\label{Vjaffard.def1}
  {\mathcal J}_{\alpha}({\mathcal V}):=\big \{ {\bf A}:=(a(i,i'))_{i, i'\in V}, \ \|{\bf A}\|_{{\mathcal J}_{\alpha}({\mathcal V})}<\infty \big \}, \end{equation}
  where
 \begin{equation}\label{Vjaffard.def2}
 \|{\bf A}\|_{{\mathcal J}_{\alpha}({\mathcal V})} :=   \sup_{i, i'\in V}(1+\rho(i,i'))^{\alpha}|a(i,i')|, \ \alpha\ge 0.
  \end{equation}
 The proof of Theorem \ref{lpstability.tm} depends highly on the following
 Wiener's lemma  for the matrix algebra ${\mathcal J}_\alpha({\mathcal V}), \alpha>d$.

\begin{Tm}\label{wienerlemma.tm}
 Let ${\mathcal V}$ be as in \eqref{v.graph} and its counting measure $\mu$
 satisfy \eqref{mupolynomial.def}.
 If ${\mathbf A}\in {\mathcal J}_\alpha({\mathcal V}), \alpha>d$,  and ${\bf A}^{-1}$ is bounded on $\ell^2$, then
 ${\mathbf A}^{-1}\in {\mathcal J}_\alpha({\mathcal V})$ too. 
\end{Tm}

 Wiener's lemma has been established for infinite matrices, pseudodifferential operators,  and integral operators
satisfying various off-diagonal decay conditions (\cite{balan08, Strohmer10, Grochenig06, Grochenig10,  GrochenigL06,
jaffard,   suncasp05,  suntams07,   sunacha08, sunca11}).
It has been shown to be crucial for  well-localization of dual Gabor/wavelet
frames,
 fast implementation in numerical
analysis,   local reconstruction in sampling theory,  local features of spatially distributed
optimization, etc. The reader may refer to the survey papers \cite{grochenig10, krishtal11} 
for historical remarks, motivation and recent advances.

The  Wiener's lemma (Theorem \ref{wienerlemma.tm}) is also used to establish the sub-optimal approximation property \eqref{rougherror} for the ``least squares" solution $\Delta_2({\bf z})$ in \eqref{d2.eqn}, for which a distributed algorithm is proposed in Section \ref{distributedalgorithm.section}.

\begin{Tm}\label{leastsqaure.tm}  
 Let  ${\mathcal G}, {\mathcal H}$ and ${\mathcal V}$ be as in Proposition \ref{linftystability.pr}.
 Assume that the sensing matrix ${\mathbf S}$ satisfies \eqref{sensingmatrix.assumption}
 and it has $\ell^2$-stability.
Then there exists a positive constant $C$ such that
\begin{equation}\label{l2estimate.in.eq}
\|\Delta_2({\bf z})-{\bf c}\|_\infty\le C \|\pmb\eta\|_\infty
\quad {\rm for \ all} \ {\bf c}, {\pmb \eta}\in \ell^\infty,
\end{equation}
where ${\bf z}={\bf S}{\bf c}+\pmb \eta$.
\end{Tm}

\section{Stability criterion for distributed sampling and reconstruction system}
\label{criterion.section}

Let ${\mathcal H}$ be the connected simple graph in \eqref{h.graph} to describe
our DSRS. Given   $\lambda'\in G$ and a positive integer $N$, define {\em truncation operators}
 $\chi_{\lambda', G}^N$ and $\chi_{\lambda', V}^N$
  by
  $$\chi_{\lambda', G}^N: \ \ell^p(G) \ni (d(\lambda))_{\lambda\in G}\longmapsto \big(d(\lambda)\chi_{B_{\mathcal H}(\lambda',N)\cap G}(\lambda)\big)_{\lambda\in G} \in \ell^p(G)$$
and
$$\chi_{\lambda', V}^N: \ \ell^p(V) \ni (c(i))_{i \in V}\longmapsto \big(c(i)\chi_{B_{\mathcal H}(\lambda',N)\cap V}(i)\big)_{i\in V}  \in \ell^p(V),$$
where $1 \le p \le \infty$ and
 $$B_{\mathcal H}(u, r):=\{v\in G\cup V,\ \rho_{\mathcal H}(u, v)\le r\}$$
 is the closed ball in ${\mathcal H}$ with center $u\in {\mathcal H}$ and radius $r\ge 0$.

For any matrix ${\mathbf A}\in {\mathcal J}_\alpha({\mathcal G}, {\mathcal V})$ with $\ell^2$-stability, we observe that
its quasi-main  submatrices $\chi_\lambda^{2N}{\mathbf A}\chi_\lambda^N, \lambda\in G$, of size  $O(N^d)$ have uniform $\ell^2$-stability
for large $N$.

\begin{Tm}\label{l2stablility.thm1}
  Let  ${\mathcal G}$ and ${\mathcal H}$
satisfy Assumptions 1 -- 4,
${\mathcal V}$ be as in \eqref{v.graph},
$\mu_{\mathcal G}$
 have polynomial growth with Beurling dimension $d$ and sampling density $D_1({\mathcal G})$, and
 let
 ${\mathbf A}\in {\mathcal J}_\alpha({\mathcal G}, {\mathcal V})$ for some $\alpha>d$. If ${\mathbf A}$
 has $\ell^2$-stability with lower bound $A \|{\mathbf A}\|_{{\mathcal J}_\alpha({\mathcal G}, {\mathcal V})}$, then
\begin{equation}\label{l2stablility.thm1.eq1}
\|\chi_{\lambda, G}^{2N}{\mathbf A}\chi_{\lambda, V}^N {\bf c}\|_2 \ge \frac{A}{2} \|{\mathbf A}\|_{{\mathcal J}_\alpha({\mathcal G}, {\mathcal V})}
\|\chi_{\lambda, V}^N{\bf c}\|_2,\ {\bf c}\in \ell^2
\end{equation}
for all $\lambda\in G$ and all integers $N$ satisfying
\begin{equation}\label{l2stablility.thm1.eq2}
 2 D_1({\mathcal G}) N^{-\alpha+d} \sqrt{L\alpha/(\alpha-d)}
 \le A.\end{equation}
 \end{Tm}

The above theorem provides a guideline to design a 
 distributed algorithm for signal reconstruction, see Section \ref{distributedalgorithm.section}.
Surprisingly, the converse of Theorem \ref{l2stablility.thm1} is true, 
   cf.  the  stability criterion in \cite[Theorem 2.1]{sunpams10}
for 
convolution-dominated 
matrices.

\begin{Tm}\label{l2stablility.thm}
 Let  ${\mathcal G}, {\mathcal H}, {\mathcal V}$ 
  be as in  Theorem \ref{l2stablility.thm1}, and
 ${\mathbf A}\in {\mathcal J}_\alpha({\mathcal G}, {\mathcal V})$ for some $\alpha>d$.
If there exist a positive constant $A_0$ and an integer $N_0\ge 3$ 
 such that
\begin{equation}\label{l2stablility.thm2.eq0}
 \hskip-0.2in A_0    \ge  4(D_0({\mathcal G}))^{2}D_1({\mathcal G})L N_0^{-\min(\alpha-d, 1)}\times  \left\{\begin{array}{ll}  \Big(\frac{4\alpha}{3(\alpha-d)}+\frac{2(\alpha-1)(\alpha-d)}{\alpha-d-1}\Big) & \hskip-0.1in {\rm if}
\ \alpha>d+1\\
\big(\frac{10(d+1)}{3}+2d\ln N_0\big) & \hskip-0.1in {\rm if} \ \alpha=d+1\\
  \Big(\frac{4\alpha}{3(\alpha-d)}+\frac{4 d}{d+1-\alpha}\Big) & \hskip-0.1in {\rm if} \ \alpha<d+1,
\end{array}\right.
\end{equation}
and for all $\lambda\in G$,
\begin{equation}\label{l2stablility.thm2.eq1}
\|\chi_{\lambda, G}^{2N_0}{\mathbf A}\chi_{\lambda, V}^{N_0} {\bf c}\|_2 \ge A_0
\|{\mathbf A}\|_{{\mathcal J}_\alpha({\mathcal G}, {\mathcal V})}
\|\chi_{\lambda, V}^{N_0}{\bf c}\|_2, \  {\bf c}\in \ell^2,
\end{equation}
then ${\mathbf A}$ has $\ell^2$-stability,
\begin{equation}\label{l2stablility.thm2.eq2}
\|{\bf A} {\bf c}\|_2\ge \frac{A_0\|{\mathbf A}\|_{{\mathcal J}_\alpha({\mathcal G}, {\mathcal V})}}{12(D_0({\mathcal G}))^{2}}\|{\bf c}\|_2, \ {\bf c}\in \ell^2.
\end{equation}
\end{Tm}

 Observe that the right hand side of \eqref{l2stablility.thm2.eq0}
could be arbitrarily small when $N_0$ is sufficiently large. This together with Theorem \ref{l2stablility.thm1}
implies that the requirements \eqref{l2stablility.thm2.eq0} and \eqref{l2stablility.thm2.eq1}
are necessary for the $\ell^2$-stability property of any matrix 
in ${\mathcal J}_{\alpha}({\mathcal G}, {\mathcal V})$. As shown in the example below, the term $N_0^{-\min(\alpha-d, 1)}$
in \eqref{l2stablility.thm2.eq0}
cannot be replaced by  $N_0^{-\beta}$  with high order $\beta>1$ even if the matrix ${\bf A}$ has finite bandwidth.
\begin{Ex} {\rm  Let ${\bf A}_0=(a_0(i-j))_{i,j\in \ZZ}$ be the bi-infinite Toeplitz matrix
with symbol $\sum_{k\in \ZZ} a_0(k)e^{-ik\xi}= 1-e^{-i\xi}$. 
Then ${\bf A}_0$ belongs to the Jaffard class ${\mathcal J}_\alpha(\ZZ, \ZZ)$ for all $\alpha\ge 0$ and it does not have $\ell^2$-stability. On the other hand,  for any $\lambda\in G=V=\ZZ$ and $N_0\ge 1$,
\begin{eqnarray*} & & \inf_{ \|\chi_{\lambda, V}^{N_0} {\bf c}\|_2=1}
\|\chi_{\lambda, G}^{2N_0}{\mathbf A}_0\chi_{\lambda, V}^{N_0} {\bf c}\|_2=
\inf_{ \|\chi_{\lambda, V}^{N_0} {\bf c}\|_2=1}
\|{\mathbf A}_0\chi_{\lambda, V}^{N_0} {\bf c}\|_2\\
& = & \inf_{|d_1|^2+\cdots+ |d_{2N_0+1}|^2=1}
\sqrt{|d_1|^2+|d_1-d_2|^2+\cdots+|d_{2N_0}-d_{2N_0+1}|^2+|d_{2N_0+1}|^2}\\
& = & 2\sin \frac{ \pi}{4N_0+4} 
\ge \frac{1}{2} N_0^{-1},
\end{eqnarray*}
where the last equality follows from \cite[Lemma 1 of Chapter 9]{kincaid}.
}\end{Ex}

For our DSRS with sensing matrix ${\bf S}$ having the polynomial off-diagonal decay  property \eqref{sensingmatrix.assumption}, the uniform stability property \eqref{l2stablility.thm2.eq1}  could be verified by finding minimal eigenvalues of its quasi-main submatrices  $ \chi_{\lambda, V}^{N_0}{\mathbf S}^T\chi_{\lambda, G}^{2N_0}{\mathbf S}\chi_{\lambda, V}^{N_0}, \lambda\in G$, of size about $O(N_0^d)$.
The  above  verification  could be implemented  on agents in the DSRS
via its computing and communication abilities.
This  provides a practical tool to verify  $\ell^2$-stability
of a DSRS and to design a robust (dynamic) DSRS against
 supplement, replacement and impairment of agents.

\section{Exponential convergence of a  distributed reconstruction algorithm}
\label{distributedalgorithm.section}

In our DSRS, 
agents could essentially catch signals not far from their locations.
So one may expect that  a signal  near any innovative position should
  substantially be determined by sampling data of neighboring agents, while
  data from  distant agents should have (almost) no influence in the reconstruction.
 The most desirable method to meet the above expectation is local exact reconstruction,
  which could be implemented in a distributed manner without iterations (\cite{Akram00, Grochenig03,  sunaicm10, sunW09}).
In such a linear  reconstruction procedure,  there is 
a left-inverse $\mathbf T$  of the sensing matrix ${\bf S}$ with finite bandwidth, 
   \begin{equation*}\label{localinverse}
  {\bf T}{\bf S}={\bf I}.
  \end{equation*}

 For our DSRS, such a left-inverse $\mathbf T$ with finite bandwidth 
may not exist
and/or it is difficult to find even it exists. We observe that $${\mathbf S}^\dag:=({\mathbf S}^T{\mathbf S})^{-1}{\mathbf S}^T$$
 is a left-inverse 
  well  approximated by matrices with finite bandwidth, and 
 \begin{equation}\label{leastsquare2.eq}
 {\bf d}_2=  {\mathbf S}^\dag{\bf z}
\end{equation}
is a suboptimal  approximation, 
where ${\bf z}$  is given in \eqref{noisysamples}.
However, it is  infeasible to find
the pseudo-inverse ${\mathbf S}^\dag$, 
 because the DSRS does not have a central processor and it has
huge amounts of agents and large number of innovative positions.
In this section, we introduce a  distributed  algorithm
to find the suboptimal approximation ${\bf d}_2$ in \eqref{leastsquare2.eq}.

\smallskip

Let ${\mathcal H}$  
be the  connected simple
graph in \eqref{h.graph}  to describe our DSRS,  
and the sensing matrix ${\mathbf S}\in {\mathcal J}_\alpha({\mathcal G}, {\mathcal V}), \alpha>d$, have $\ell^2$-stability.
Then ${\bf d}_2$ in \eqref{leastsquare2.eq}
 is the unique solution to
the ``normal'' equation
 \begin{equation}\label{normal.eqn}
 {\mathbf S}^T{\mathbf S} {\bf d}_2={\bf S}^T{\bf z}.
 \end{equation}
As
 principal submatrices
$\chi_{\lambda, V}^{N}{\mathbf S}^T {\mathbf S} \chi_{\lambda, V}^{N}$
of  the positive definite matrix ${\bf S}^T{\bf S}$ are uniformly stable, 
 we solve localized
linear systems
\begin{equation}\label{localinverse.eqn}
\chi_{\lambda, V}^{N}{\mathbf S}^T {\mathbf S} \chi_{\lambda, V}^{N}
{\bf d}_{\lambda,N}= \chi_{\lambda, V}^{N} {\bf S}^T {\bf z},\ \lambda\in G,\end{equation}
 of size  $O(N^d)$,  whose   solutions  ${\bf d}_{\lambda, N}$ are supported in the ball $B_{\mathcal{H}}(\lambda, N)\cap V$.
 One of {\bf crucial} results of this paper is that for large integer $N$,
 the solution ${\bf d}_{\lambda, N}$ provides a reasonable approximation of
the ``least squares" solution ${\bf d}_{2}$ 
 inside the half ball $B_{\mathcal{H}}(\lambda, N/2)\cap V$, 
  see \eqref{convergence.prop.eq2}
 in Proposition \ref{convergence.prop}.
However, the above local approximation 
can not be implemented distributedly in the DSRS, as only agents
 on the graph ${\mathcal G}$ have computing and telecommunication ability.
So we propose to compute
\begin{equation}\label{wlambdan.def}
 {\bf w}_{\lambda, N} := \chi_{\lambda, G}^N {\bf S}
 \chi_{\lambda, V}^N
 (\chi_{\lambda, V}^{N}{\mathbf S}^T {\mathbf S} \chi_{\lambda, V}^{N})^{-1}
 {\bf d}_{\lambda, N}
  =  \chi_{\lambda, G}^N {\bf S}
 \chi_{\lambda, V}^N
 (\chi_{\lambda, V}^{N}{\mathbf S}^T {\mathbf S} \chi_{\lambda, V}^{N})^{-2} \chi_{\lambda, V}^N {\bf S}^T {\bf z}\end{equation}
instead, which approximates 
\begin{equation}\label{wls.def}
{\bf w}_{LS}:= {\bf S} ({\bf S}^T{\bf S})^{-1} {\bf d}_2\end{equation}
 inside $B_{\mathcal{G}}(\lambda, N/2)\cap G$, see \eqref{convergence.prop.eq3}
 in the proposition below.

\begin{pr} \label{convergence.prop}
Let  ${\mathcal G}$ and ${\mathcal H}$
satisfy Assumptions 1 -- 4,
${\mathcal V}$ be as in \eqref{v.graph}, and let the sensing matrix 
 ${\mathbf S}\in {\mathcal J}_\alpha({\mathcal G}, {\mathcal V}), \alpha>d$,
  have $\ell^2$-stability with lower stability bound $A \|{\bf S}\|_{{\mathcal J}_\alpha({\mathcal G}, {\mathcal V})}$. Take an integer $N$ satisfying \eqref{l2stablility.thm1.eq2}, and
   set
 $$\theta=\frac{2\alpha-2d}{2\alpha-d}\in (0,1)\ \ {\rm and} \ \ r_0=1-\frac{A^2(\alpha-d)^2}{2^{\alpha+1}D_1D_1(\mathcal{G})\alpha^2}.$$
Then
\begin{equation}\label{convergence.prop.eq2}
\hskip-0.02in
\|\chi_{\lambda, V}^{N/2} ({\bf d}_{\lambda,N}-{\bf d}_{2})\|_{\infty} \le
D_3 (N+1)^{-\alpha+d} \|{\bf d}_2\|_\infty
\end{equation}
 and
\begin{equation} \label{convergence.prop.eq3}
\hskip-0.02in
 \|\chi_{\lambda, G}^{N/2}  ({\bf w}_{\lambda, N}- {\bf w}_{LS})\|_\infty\le
D_4 (N+1)^{-\alpha+d} \|{\bf d}_2\|_\infty,
\end{equation}
where
$D_3=\frac{2^{2\alpha-d+1}\alpha D_1D_2}{\alpha-d}$, $D_4=\big(\frac{2^{3\alpha-d+3}\alpha L^2D_1(\mathcal{G})D_2^2}{\alpha-d}+LD_2\big)\|\mathbf{S}\|_{\mathcal{J}_\alpha(\mathcal{G},\mathcal{V})}^{-1}$,
 and
\begin{equation}\label{convergence.prop.eq4}
D_2=\sum \limits_{n=0}^{\infty}\Big(\frac{2^{2\alpha+d/2+4}D_1^3\alpha^2}{r_0^{1-\theta}(\alpha-d)^2}\Big)^{\frac{2-\theta}{(1-\theta)^2}n^{\log_2^{(2-\theta)}}}r_0^n.
\end{equation}
\end{pr}

Take  a maximal $\frac{N}4$-disjoint subset $G_{N/4}\subset G$ satisfying \eqref{vn.def1} and \eqref{vn.def2}.
We patch
  ${\bf w}_{\lambda,N}, \lambda\in G_{N/4}$, in \eqref{wlambdan.def}
 together to generate a linear approximation
\begin{equation}\label{solutionpataching}
{\bf w}_N^*=\sum_{\lambda\in G_{N/4}} {\pmb \Theta}_{\lambda, N} \chi_{\lambda, G}^{N/2} {\bf w}_{\lambda,N}\end{equation}
of the bounded vector  ${\bf w}_{LS}$,
where ${\pmb \Theta}_{\lambda, N}$ is a diagonal matrix with diagonal  entries
$${\theta}_{\lambda, N}(\lambda^{\prime\prime}) = \frac{\chi_{B_{\mathcal G}(\lambda,N/2)}(\lambda^{\prime\prime})} {\sum_{\lambda'\in G_{N/4}} \chi_{B_{\mathcal G}(\lambda',N/2)}(\lambda^{\prime\prime})}, \ \lambda^{\prime\prime}\in G.
$$
The above approximation is well-defined
as 
$ \{B_{\mathcal G}(\lambda', N/2), \ \lambda'\in G_{N/4}\}$ is a finite covering of $G$
by  \eqref{hgball.eq} and Proposition \ref{covering.pr}.
Moreover, we obtain from Proposition \ref{convergence.prop} that
\begin{eqnarray}\label{rxerror}
 \|{\bf w}^*_N-{\bf w}_{LS}\|_\infty
&\hskip-0.08in = & \hskip-0.08in
\Big\|\sum_{\lambda\in G_{N/4}} {\pmb \Theta}_{\lambda, N}
\chi_{\lambda, G}^{N/2}({\bf w}_{\lambda,N}-{\bf w}_{LS})\Big\|_\infty \nonumber\\
&\hskip-0.08in \le & \hskip-0.08in \sup_{\lambda^{\prime\prime} \in G}
\sum_{\lambda\in G_{N/4}} \theta_{\lambda, N}(\lambda^{\prime\prime})
   \| \chi_{\lambda, G}^{N/2}({\bf w}_{\lambda,N}-{\bf w}_{LS})\|_\infty\nonumber\\
  &\hskip-0.08in \le & \hskip-0.08in  D_4 (N+1)^{-\alpha+d}
\|{\bf d}_2\|_\infty.
\end{eqnarray}
Therefore, the  moving  consensus ${\bf w}^*_N$  of  ${\bf w}_{\lambda,N}, \lambda\in G_{N/4}$,
 provides a good approximation to ${\bf w}_{LS}$ in \eqref{wls.def} for large $N$.
In addition, ${\bf w}^*_N$ depends on the observation ${\bf z}$ linearly,
\begin{equation}\label{wN*.def}
{\bf w}_N^*= {\bf R}_N {\bf S}^T {\bf z}
\end{equation}
for some matrix ${\bf R}_N$ with bandwidth $2N$
and
\begin{equation}\label{con.4}
\|{\bf R}_N\|_{{\mathcal J}_\alpha({\mathcal G}, {\mathcal V})}\le D_5:=\frac{(\alpha-d)^2LD_2^2}{\alpha^2D_1D_1(\mathcal{G})\|\mathbf{S}\|_{{\mathcal J}_\alpha({\mathcal G}, {\mathcal V})}^3}.
\end{equation}

Given noisy samples ${\bf z}$, we may use
${\bf w}_N^*$ in \eqref{wN*.def} as the first approximation of ${\bf w}_{LS}$,
\begin{equation}\label{wn.def0}
{\bf w}_1={\bf R}_N {\bf S}^T {\bf z}\end{equation}
and recursively define
\begin{equation}\label{wn.def}
{\bf w}_{n+1}= {\bf w}_n+{\bf w}_1-{\bf R}_N {\bf S}^T {\bf S} {\bf S}^T {\bf w}_n,
\ n\ge 1.
\end{equation}
In the next theorem, we show that the above sequence ${\bf w}_n, n\ge 1$, converges  exponentially to some bounded vector ${\bf w}$, not necessarily ${\bf w}_{LS}$,
satisfying the consistent condition \begin{equation}\label{con.6}
{\bf S}^T {\bf w}={\bf S}^T {\bf w}_{LS}  ={\bf d}_2.
\end{equation}

\begin{Tm}\label{convergence.thm}
Let  ${\mathcal G},{\mathcal H}$  and ${\mathcal V}$
be as in Proposition \ref{convergence.prop}, and let ${\bf w}_n, n\ge 1$, be as in  \eqref{wn.def0} and \eqref{wn.def}.
 Suppose that $N$ satisfies \eqref{l2stablility.thm1.eq2} and
 \begin{equation}\label{anotherrequirementonN}
r_1:=  \frac{D_1(\mathcal{G})D_4L\alpha}{\alpha-d} \|{\bf S}\|_{{\mathcal J}_\alpha(\mathcal{G},\mathcal{V})} (N+1)^{-\alpha+d}<1.\end{equation}
Set $$D_6=\frac{2^{2\alpha+2}\alpha L^3(D_1(\mathcal{G}))^2D_2^2}{(\alpha-d)(1-r_1)D_1\|{\bf S}\|_{{\mathcal J}_\alpha(\mathcal{G},\mathcal{V})}}.$$
Then ${\bf w}_n$ and  ${\bf S}^T{\bf w}_n, n\ge 1$,
converge exponentially to a bounded vector ${\bf w}$ in \eqref{con.6} and the ``least squares" solution ${\bf d}_2$ in \eqref{leastsquare2.eq} respectively,
\begin{equation}\label{convergence.thm.eq1}
\|{\bf w}_n-{\bf w}\|_\infty \le D_6 r_1^n \|{\bf d}_{2}\|_\infty
\end{equation}
and
\begin{equation}\label{convergence.thm.eq2}
\|{\bf S}^T{\bf w}_n-{\bf d}_2\|_\infty \le \frac{D_1(\mathcal{G})D_6L\alpha }{\alpha-d}\|\mathbf{S}\|_{{\mathcal J}_\alpha({\mathcal G}, {\mathcal V})}r_1^n \|{\bf d}_{2}\|_\infty, \  n\ge 1.
\end{equation}
\end{Tm}

By the above theorem, each agent  should have minimal storage,  computing, and telecommunication capabilities.
 Furthermore, the algorithm \eqref{wn.def0} and \eqref{wn.def} will have faster convergence (hence less delay for signal reconstruction) by selecting large $N$ when agents  have
larger storage, more computing power,  and higher telecommunication capabilities.  In addition,
no iteration is needed for sufficiently large $N$,
 and the reconstructed signal is  approximately to the one obtained by the finite-section method,
 cf.  \cite{Christensen05} and  simulations in  Section \ref{simulations.section}.

The iterative algorithm \eqref{wn.def0} and \eqref{wn.def}
can be recast as follows:
\begin{equation}\label{wn.def1}
{\bf w}_1={\bf R}_N {\bf S}^T {\bf z} \ \  {\rm and}\ \
{\bf e}_1={\bf w}_1-{\bf R}_N {\bf S}^T {\bf S} {\bf S}^T {\bf w}_1,
\end{equation}
and
\begin{equation}\label{wn.def2}
\left\{\begin{array}{l}
{\bf w}_{n+1}= {\bf w}_n+ {\bf e}_n\\
{\bf e}_{n+1}= {\bf e}_n-{\bf R}_N {\bf S}^T {\bf S} {\bf S}^T {\bf e}_n,
\ n\ge 1.
\end{array} \right.
\end{equation}
Next, we present a  distributed  implementation of
the  algorithm \eqref{wn.def1} and \eqref{wn.def2}  when ${\bf S}$  has bandwidth $s$.
Select a threshold $\epsilon$ and an integer $N\ge s$ satisfying \eqref{anotherrequirementonN}.
Write
$$
\left\{\begin{array}{l} {\bf S}^T=({\bf a}(i, \lambda))_{i\in V, \lambda\in G}\\
{\bf R}_N {\bf S}^T=({\bf b}_N(\lambda, \lambda'))_{\lambda, \lambda'\in G}\\
{\bf R}_N {\bf S}^T {\bf S} {\bf S}^T= ({\bf c}_N(\lambda, \lambda'))_{\lambda, \lambda'\in G} \\
{\bf z}=({\bf z}(\lambda))_{\lambda\in G},
\end{array}\right.
$$
and
$${\bf w}_n=({\bf w}_n(\lambda))_{\lambda\in G} \ \ {\rm  and}\  \
{\bf e}_n=({\bf e}_n(\lambda))_{\lambda\in G},\  n\ge 1. $$
 We assume that any agent $\lambda\in G$  stores vectors  ${\bf a}(i, \lambda'), {\bf b}_N(\lambda, \lambda'), {\bf c}_N(\lambda, \lambda')$ and ${\bf z}(\lambda')$, where $(i,\lambda)\in T$ and $\lambda'\in B_{\mathcal G}(\lambda, 2N+3s)$.
The following is the distributed implementation of the  algorithm \eqref{wn.def1} and \eqref{wn.def2} for an agent $\lambda\in G$.

\smallskip

\noindent {\bf Distributed algorithm \eqref{wn.def1} and \eqref{wn.def2} for signal reconstruction:}
\begin{itemize}
\item[{1.}]  Input ${\bf a}(i, \lambda'), {\bf b}_N(\lambda, \lambda'), {\bf c}_N(\lambda, \lambda')$ and ${\bf z}(\lambda')$, where $(i,\lambda)\in T$ and $\lambda'\in B_{\mathcal G}(\lambda, 2N+3s)$.

\item[{2.}] Input stop criterion $\epsilon>0$ and maximal number of iteration steps $K$.

\item[{3.}]  Compute
${\bf w}(\lambda) = \sum_{\lambda'\in B_{\mathcal G}(\lambda, 2N+s)} {\bf b}_N(\lambda, \lambda') {\bf z}(\lambda')$.

\item[{4.}] Communicate with neighboring agents in $B_{\mathcal G}(\lambda, 2N+3s)$   to obtain
data 
${\bf w}(\lambda')$, $\lambda'\in B_{\mathcal G}(\lambda, 2N+3s)$.

\item[{5.}] Evaluate the sampling error term
${\bf e}(\lambda) = {\bf w}(\lambda)-\sum_{\lambda'\in B_{\mathcal G}(\lambda, 2N+3s)} {\bf c}_N(\lambda, \lambda') {\bf w}(\lambda')$.

    \item[{6.}] Communicate with neighboring agents  in $B_{\mathcal G}(\lambda, 2N+3s)$   to obtain error
data 
${\bf e}(\lambda')$, $\lambda'\in B_{\mathcal G}(\lambda, 2N+3s)$.


\item[{7.}] {\bf for} $n=2$ to $K$ {\bf do}

\begin{itemize}
\item[{7a.}] Compute $\delta=\max_{\lambda'\in B_{\mathcal G}(\lambda, 2N+3s)}|{\bf e}(\lambda')|$.

\item[{7b.}] {\bf Stop} if $\delta\le \epsilon$, else  {\bf do}

\item[{7c.}] Update ${\bf w}(\lambda)={\bf w}(\lambda)+{\bf e}(\lambda)$.

\item[{7d.}] Update
${\bf e}(\lambda) = {\bf e}(\lambda)-\sum_{\lambda'\in B_{\mathcal G}(\lambda, 2N+3s)} {\bf c}_N(\lambda, \lambda') {\bf e}(\lambda')$.

\item[{7e.}] Communicate with neighboring agents located in $B_{\mathcal G}(\lambda, 2N+3s)$   to obtain error
data 
${\bf e}(\lambda')$, $\lambda'\in B_{\mathcal G}(\lambda, 2N+3s)$.

\end{itemize}

\item[{}] {\bf end for}

\end{itemize}

We conclude this section by discussing the complexity of 
the  distributed algorithm \eqref{wn.def1} and \eqref{wn.def2}, which depends essentially on $N$.
In its implementation, the data storage requirement  for each agent is about
$(L+3)(2N+3s+1)^d$.
In each iteration, 
the computational cost for each agent is about $O(N^d)$ mainly used for updating the error ${\bf e}$.
The communication cost for each agent is about $O(N^{d+\beta})$ if
the communication between distant agents $\lambda, \lambda'\in G$, processed through  their shortest path, has its cost being proportional to  $(\rho_{\mathcal G}(\lambda, \lambda'))^\beta$ for some $\beta\ge 1$.
By Theorem \ref{convergence.thm}, the number of iteration steps needed to reach the accuracy $\epsilon$
is about $O(\ln (1/\epsilon)/\ln N)$.
Therefore the total computational  and communication cost for each agent  are
about  $O(\ln (1/\epsilon) N^d/\ln N)$ and $O(\ln (1/\epsilon) N^{d+\beta}/\ln N)$, respectively.

\section{Numerical simulations}
\label{simulations.section}

In this section, we present two simulations to demonstrate the distributed algorithm \eqref{wn.def1} and \eqref{wn.def2} for stable signal reconstruction.

 Agents in the first simulation are almost uniformly deployed on  the circle of radius $R/5$, and  their locations are at
  $$\lambda_l:=\frac{R}{5}\Big(\cos \frac{2\pi \theta_l}{R}, \sin \frac{2\pi \theta_l}{R}\Big),\ 1\le l\le R,$$
   where $R\ge 1$ and
$\theta_l\in l+[-1/4, 1/4]$  are randomly selected.
Every agent in the SDS
 has a direct communication channel to its two adjacent agents. Then the  graph ${\mathcal G}_c=
 (G_c, S_c)$ to describe
 the SDS   
  is  a cycle graph,  where $G_c=\{\lambda_1, \ldots, \lambda_R\}$
 and $S_c = \big\{(\lambda_1, \lambda_2), \ldots, (\lambda_{R-1}, \lambda_R)$, $(\lambda_R, \lambda_1)$, $(\lambda_1, \lambda_R)$, $(\lambda_R, \lambda_{R-1}), \ldots, (\lambda_2, \lambda_1)\big\}$. 
Take innovative positions
$${\bf p}_i:=r_i \Big(\cos \frac{2\pi i}{R},\ \sin \frac{2\pi i}{R}\Big), \ 1\le i\le R,$$
deployed almost uniformly near the circle of radius $R/5$, where
 $r_i\in R/5+[-1/4, 1/4]$ are randomly selected.
Given any innovative position ${\bf p}_i, 1\le i\le R$,
it has  three   anchor agents
$\lambda_i, \lambda_{i-1}$ and $\lambda_{i+1}$, where  $\lambda_0=\lambda_R$ and $\lambda_{R+1}=\lambda_1$.
Set $V_c=\{{\bf p}_i, 1\le i\le R\}$ and
$T_c=\{({\bf p}_i, \lambda_{i-j}),  i=1, \ldots,  R \ {\rm and} \ j=0, \pm 1\}$.
Then  ${\mathcal H}_c=(G_c\cap V_c,  S_c\cup T_c\cup T_c^*)$ is the 
graph
to describe the DSRS, see   Figure \ref{circle.fig}.

  \begin{figure}[h]
\begin{center}
\includegraphics[width=84mm, height=64mm]{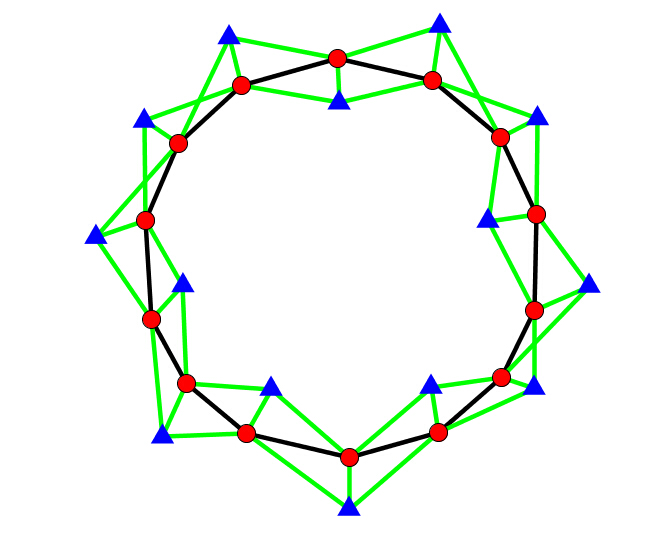}
\caption{The  graph ${\mathcal H}_c=(G_c\cap V_c,  S_c\cup T_c\cup T_c^*)$
  to describe  the DSRS in the first  simulation, where
  vertices in $G_c$, edges in $S_c$,  vertices in  $V_c$ and edges in $T_c \cup T_c^{\ast}$
   are plotted in red circles, black lines,  blue triangles and green lines,
   respectively.
 }
\label{circle.fig}
\end{center}
\end{figure}

Let $\varphi({\bf t}):= \exp (-(t_1^2+ t_2^2)/2)$ for ${\bf t}=(t_1, t_2)$. 
Gaussian signals
  \begin{equation} \label{simulation1.signal} f({\bf t})=\sum_{i=1}^R c(i) \varphi ({\bf t}-{\bf p}_i) \end{equation}
  to be  sampled and reconstructed have their
amplitudes $c(i)\in [0, 1]$ being randomly chosen,
see the left image of Figure \ref{circlesignalreconstruction.fig}.
 \begin{figure}[h]
\begin{center}
\includegraphics[width=58mm, height=48mm]{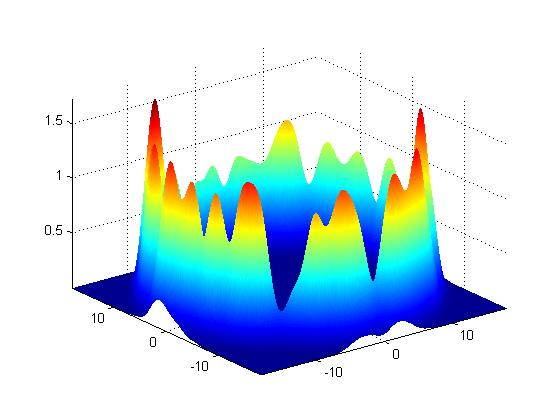}
\includegraphics[width=58mm, height=48mm]{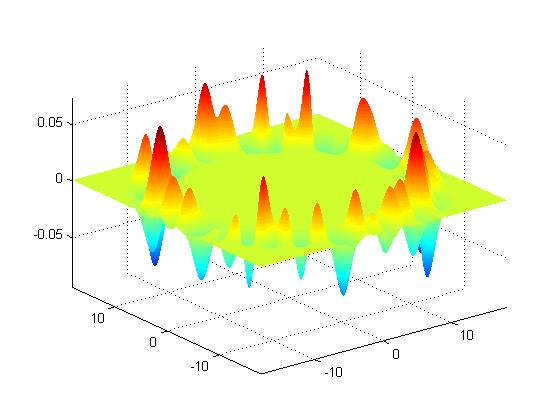}
\caption{Plotted on the left is the signal $f$ in \eqref{simulation1.signal} with  $R=80$.
On the right is
 the difference between the  signal $f$ 
  and  the  reconstructed signal $f_{n ,N, \delta}$ 
 with $n=10, N=6$ and $\delta=0.05$.
 }
\label{circlesignalreconstruction.fig}
\end{center}
\end{figure}
In the first simulation,   we consider ideal sampling procedure. Thus for the agent $\lambda_l, 1\le l\le R$,
the noisy sampling data  acquired is
\begin{equation}\label{noisysamples.firstsimulation}
y_\delta(l)=\sum_{i=1}^R c(i) \varphi(\lambda_l-{\bf p}_i)+\eta(l),\end{equation}
where $\eta(l)\in [-\delta, \delta]$ are randomly generated with bounded noise level $\delta>0$.

Our first simulation shows that the distributed algorithm \eqref{wn.def1} and \eqref{wn.def2}
  converges for $N\ge 5$ and the convergence rate is almost independent
of the network size  $R$, cf. the upper bound estimate  in \eqref{convergence.thm.eq2}. 


Let
$f_{n, N, \delta}({\bf t}):=\sum_{i=1}^R c_{n, N, \delta}(i) \varphi ({\bf t}-{\bf p}_i)$ 
 be the reconstructed signal in the $n$-th iteration
 by applying the distributed algorithm \eqref{wn.def1} and \eqref{wn.def2}
 from the noisy sampling data in \eqref{noisysamples.firstsimulation}, see the right image of Figure  \ref{circlesignalreconstruction.fig}. Define maximal reconstruction errors
  $$\epsilon(n, N, \delta):=
  \left\{\begin{array}{ll}
  \max_{1\le i\le R} |c(i)|  & {\rm if} \ n=0,\\
 \max_{1\le i\le R} |c_{n, N, \delta}(i)-c(i)| & {\rm if} \ n\ge 1.
 \end{array}\right. $$
  Presented in
  Table \ref{iterationerror.tab}
 is the average of reconstruction errors $\epsilon(n, N, \delta)$ with
 500 trials in noiseless environment ($\delta=0$), where the network size $R$ is 80.
 It indicates that the proposed distributed algorithm \eqref{wn.def1} and \eqref{wn.def2}
 has faster convergence rate for larger $N\ge 5$,  and we only need three iteration steps to have a near perfect reconstruction from its noiseless samples when $N=10$.
 \begin{table}[h]
\caption{Maximal reconstruction errors $\epsilon(n, N, \delta)$ with $\delta=0$
}
\begin{tabular}{c|c|c|c|c|c|c}
\hline\hline
 \backslashbox {\ n}{N \ }  

 &5
 &  6 
& 7 
 & 8 
 & 9
 & 10\\
 \hline
0 & 0.9874 &0.9881 &  0.9878  &  0.9876 & 0.9877 & 0.9884 \\
\hline
1 & 0.9875 &0.4463 &  0.3073  &  0.1940  & 0.1055 & 0.0523 \\
\hline
2 & 0.6626 &0.2046 &  0.0794 &  0.0271  & 0.0124  &0.0024 \\
\hline
3 &0.3624 & 0.0926 &  0.0240 &  0.0045  & 0.0014 & 0.0001 \\
\hline
4 & 0.2535 &0.0443&  0.0068 &  0.0006  & 0.0001 & 0.0000 \\
\hline
5 &0.1742 & 0.0206 &  0.0018 &  0.0001  & 0.0000 & 0.0000 \\
\hline
6 & 0.1169 &0.0093 &  0.0005 &  0.0000  & 0.0000 & 0.0000 \\
\hline
7 & 0.0840 &0.0042 &  0.0001 &  0.0000  & 0.0000 & 0.0000 \\
\hline
8 & 0.0579 &0.0017 &  0.0000 &  0.0000  & 0.0000 & 0.0000 \\
\hline
9 & 0.0411 &0.0007 &  0.0000 &  0.0000  & 0.0000 & 0.0000 \\
\hline
10 & 0.0289 &0.0003 &  0.0000 &  0.0000  & 0.0000 & 0.0000 \\
%
\hline \hline
\end{tabular}
\label{iterationerror.tab}
\end{table}

The robustness of the proposed algorithm \eqref{wn.def1} and \eqref{wn.def2} against sampling noises is tested
and confirmed,
see Figure \ref{circlesignalreconstruction.fig}.
Moreover, it is observed that the maximal reconstruction error
$\epsilon(n, N, \delta)$ with large $n$ depends almost linearly on the noise level $\delta$, cf. the
sub-optimal approximation property 
 in Theorem \ref{leastsqaure.tm}.

\smallskip

In the next simulation, agents are  uniformly deployed on two concentric circles
  and each agent  has   direct communication channels to its three adjacent agents.
  Then the  graph
  ${\mathcal G}_p= (G_p, S_p)$
 to describe
 our SDS  is  a  prism graph with vertices having physical locations,  
   $$\mu_l:=\left\{\begin{array}
   {ll} \frac{R}{10}\big(\cos \frac{4\pi \theta_l}{R}, \sin \frac{4\pi \theta_l}{R}\big)  & \ {\rm if} \ \ 1\le l\le \frac{R}{2}\\
  \big(\frac{R}{10}+1\big)\big(\cos \frac{4\pi \theta_l}{R}, \sin \frac{4\pi \theta_l}{R}\big) &  \ {\rm if}\ \ \frac{R}{2}+1\le l\le R,
  \end{array}\right.$$
where $R\ge 2$ and
$\theta_l\in l+[-1/4, 1/4], 1\le l\le R$, are randomly selected.
The innovative positions
$${\bf q}_i:=r_i \big(\cos \frac{4\pi i}{R},\ \sin \frac{4\pi i}{R}\big), \ 1 \le i\le \frac{R}{2},$$
 have  four  anchor agents
$\mu_i, \mu_{i+1}$, $\mu_{i+R/2}$ and $\mu_{i+R/2+1}$, where  $\mu_0=\mu_{R/2}$, $\mu_{R+1}=\mu_{R/2+1}$,  and $r_i\in \frac{R}{10}+[\frac14, \frac34]$
are randomly selected.
Set $V_p=\{{\bf q}_i, 1\le i\le \frac{R}{2}\}$ and
$T_p=\{({\bf q}_i, \mu_{i+j}),  i=1, \ldots,  \frac{R}{2} \ {\rm and} \ j=0, 1, \frac{R}{2}, \frac{R}{2}+1\}$. Thus the graph
 ${\mathcal H}_p=(G_p\cap V_p,  S_p\cup T_p\cup T_p^*)$  to describe our DSRS
 is a connected simple graph,  see the left image of
Figure \ref{ring.fig}. 
\begin{figure}[h]
\begin{center}
\includegraphics[width=62mm, height=48mm]{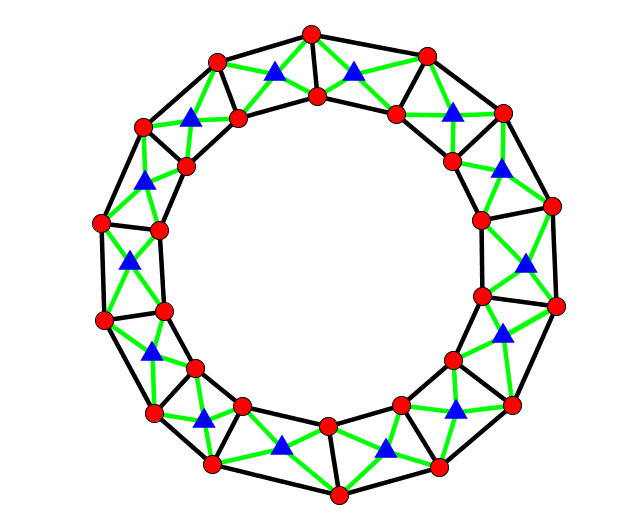}
\includegraphics[width=68mm, height=48mm]{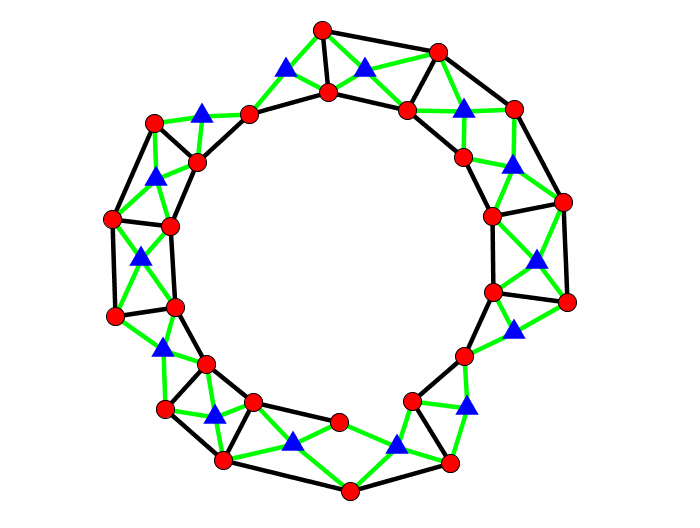}
\caption{Plotted on the left is the  graph ${\mathcal H}_p=(G_p\cap V_p,  S_p\cup T_p\cup T_p^*)$ to describe  the DSRS, where  vertices in $G_p$ and $V_p$  are  in red circles and blue triangles,
and edges in $S_p$ and $T_p \cup T_p^{\ast}$ are  in black solid lines and green solid lines,  respectively.
On the right is  a subgraph of ${\mathcal H}_p$, where some agents are completely dysfunctional and
some  have communication channels to one or two of their nearby agents clogged.
 }
\label{ring.fig}
\end{center}
\end{figure}

Following the first simulation, we consider the ideal sampling procedure of signals,  
  \begin{equation} \label{simulation2.signal}
  g({\bf t})=\sum_{i=1}^{R/2} d(i) \varphi({\bf t}-{\bf q}_i),
  \end{equation}
where  $d(i)\in [0, 1], 1\le i\le R/2$,
are randomly selected, see the left image of Figure \ref{ringsignalreconstruction.fig}.
 \begin{figure}[h]
\begin{center}
\includegraphics[width=58mm, height=48mm]{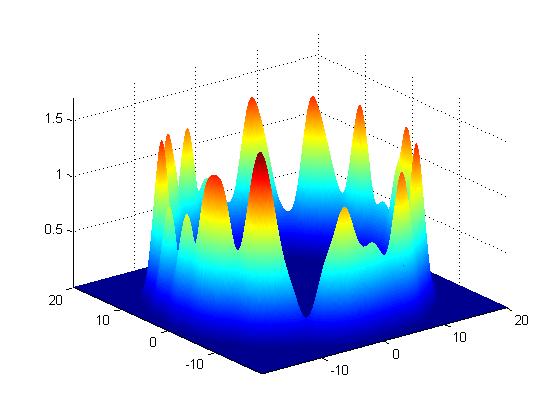}
\includegraphics[width=58mm, height=48mm]{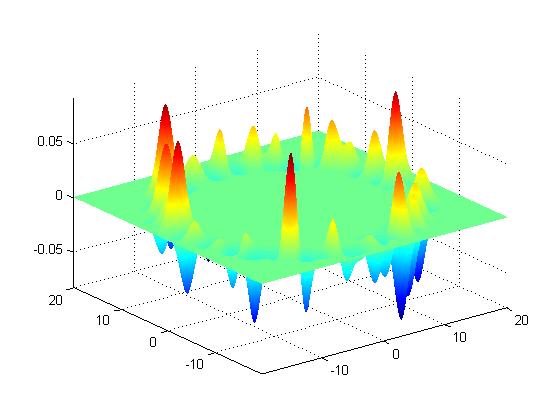}
\caption{Plotted on the left is the signal $g$ in \eqref{simulation2.signal}
with  $R=160$.
On the right is
  the difference between the  signal $g$
  and its approximation $g_{n,N,\delta}$,
  where $n=4, N=6, \delta=0.05$, and
    agents located at $\mu_1, \mu_{87}$ are completely dysfunctional, while agents located at
  $\mu_{11}, \mu_{51}, \mu_{91}$ have their partial communication channels clogged.
 }
\label{ringsignalreconstruction.fig}
\end{center}
\end{figure}
Then the noisy sampling data  acquired by the agent $\mu_l, 1\le l\le R$,
 is
\begin{equation}\label{noisysamples.secondsimulation}
y_\delta(l)=\sum_{i=1}^{R/2} d(i) \varphi(\mu_l-{\bf q}_i)+\eta(l),\end{equation}
where $\eta(l)\in [-\delta, \delta]$ are randomly selected with bounded noise level $\delta>0$.
Applying the distributed algorithm \eqref{wn.def1} and \eqref{wn.def2},
 we obtain approximations
\begin{equation}\label{simulation2.reconstructedsignal}
g_{n, N, \delta}({\bf t})=\sum_{i=1}^{R/2} d_{n, N, \delta}(i) \varphi ({\bf t}-{\bf q}_i),\ n\ge 1,\end{equation}
 of the signal $g$ in \eqref{simulation2.signal}.
Our simulations illustrate that
the distributed algorithm \eqref{wn.def1} and \eqref{wn.def2}
   converges for $N\ge 3$
   and the signal $g$ can
be reconstructed near perfectly from its noiseless samples
in 12 steps for $N=3$, 7 steps for $N=4$, 5 steps for $N=5$, 4 steps for $N=6$, and 3 steps for $N=7$, cf. Table \ref{iterationerror.tab} in the first simulation.

The robustness of the  proposed distributed algorithm \eqref{wn.def1} and \eqref{wn.def2} against sampling noises
and dysfunctions of agents in the DSRS is tested
and confirmed,
see   the right graph of Figure
\ref{ring.fig} and the right image of Figure \ref{ringsignalreconstruction.fig}.

\section{Proofs}
\label{proofs.section}

In this section, we include proofs of Propositions \ref{covering.pr}, \ref{shortestpath.pr},
\ref{Vdistance.pr}, \ref{Vdoublingmeasure.pr}, \ref{Vpolynomial.pr}, \ref{jaffard.pr},
\ref{bandapproximation.pr}, \ref{convergence.prop}, and Theorems
\ref{lpstability.tm},  \ref{wienerlemma.tm},  \ref{leastsqaure.tm}, \ref{l2stablility.thm1}, \ref{l2stablility.thm},
 \ref{convergence.thm}.

\subsection{Proof of Proposition \ref{covering.pr}}
\label{coveringpf.appendix}

  For any $\lambda \in G$, take $\lambda_m \in G_N$ with $B_{\mathcal G}(\lambda, N)\cap B_{\mathcal G}(\lambda_m, N)\neq \emptyset$.
 Then
  $$\rho_{\mathcal G}(\lambda,\lambda_m) \le \rho_{\mathcal G}(\lambda,\lambda')+\rho_{\mathcal G}(\lambda',\lambda_m)\le 2N,$$
 where
 $\lambda'$ is a vertex in $B_{\mathcal G}(\lambda, N)\cap B_{\mathcal G}(\lambda_m, N)$.
This proves that for any $N'\ge 2N$,  balls $\{B_{\mathcal G}(\lambda_m, N'), \lambda_m\in G_N\}$    provide a  covering for  $G$,
\begin{equation}
\label{covering.pr.pf.eq0}
G\subset \bigcup_{\lambda_m\in G_N}B_{\mathcal G}(\lambda_m,N'),\end{equation}
and hence the first inequality in \eqref{covering.pr.eq1} follows.

Now we prove the last inequality in \eqref{covering.pr.eq1}. Take $\lambda\in G$. For  any $\lambda_m, \lambda_{m'} \in G_N \cap B_{\mathcal G}(\lambda,N')$,
$$
\rho_{\mathcal G}(\lambda',\lambda_{m'})\le \rho_{\mathcal G}(\lambda',\lambda_m)+\rho_{\mathcal G}(\lambda_m,\lambda)+\rho_{\mathcal G}(\lambda,\lambda_{m'})\le 
 2N'+N
$$
for all $\lambda'\in B(\lambda_m,N)$,  which implies that
\begin{equation}\label{covering.pr.pf.eq1}
B_{\mathcal G}(\lambda_m,N)\subset B_{\mathcal G}(\lambda_{m'},2N'+N).\end{equation}
Hence
\begin{eqnarray} \label{BNestimate.eq}
  \sum_{\lambda_m \in G_N}\chi_{B_{\mathcal G}(\lambda_m, N')}(\lambda) 
 \hskip-.1in  & \le \hskip-.1in&
\frac{\mu_{\mathcal G}(\cup_{\lambda_m\in G_N\cap B_{\mathcal G}(\lambda,N')} B_{\mathcal G}(\lambda_m,N))}{\inf_{\lambda_m\in G_N\cap B_{\mathcal G}(\lambda,N')} \mu_{\mathcal G}(B_{\mathcal G}(\lambda_m,N))}\nonumber \\
\hskip-.1in& \le \hskip-.1in& \sup_{\lambda_m\in G_N\cap B_{\mathcal G}(\lambda,N')} \frac{
\mu_{\mathcal G}(B_{\mathcal G}(\lambda_m, 2N'+N))}{\mu_{\mathcal G}(B_{\mathcal G}(\lambda_m,N))}\le  (D_0({\mathcal G}))^{\lceil\log_2 (2N'/N+1)\rceil},
\end{eqnarray}
where the first inequality holds as $B_{\mathcal G}(\lambda_m,N), \lambda_m\in V_N$, are disjoint,
the second one is true by \eqref{covering.pr.pf.eq1},
and the third inequality follows from the doubling assumption
 \eqref{doublingconstants.def}. 


\subsection{Proof of Proposition \ref{shortestpath.pr}}
\label{shortestpathpr.appendix}

 By the structure of the graph ${\mathcal H}$,
 it suffices to show that the shortest path in ${\mathcal H}$ to connect distinct vertices $\lambda, \lambda'\in G$
must be a path in its subgraph ${\mathcal G}$.  Suppose on the contrary that
$\lambda u_1\cdots u_{k-1} u_k u_{k+1}\cdots u_n\lambda'$ is a shortest path in ${\mathcal H}$ of length $\rho_{\mathcal H} (\lambda, \lambda')$ with vertex $u_k$ along the path belonging to $V$.
Then $u_{k-1}$ and $u_{k+1}$ are anchor agents of $u_k$ in $G$.

  For the case that $u_{k-1}$ and  $u_{k+1}$ are distinct anchor agents of the innovative position $u_k$, $(u_{k-1}, u_{k+1})\in S$ by  \eqref{basicassumption1}. Hence
$\lambda u_1\cdots u_{k-1}u_{k+1}\cdots u_n\lambda'$ is a path  of length
$\rho_{\mathcal H}(\lambda, \lambda')-1$ to connect vertices $\lambda$ and $\lambda'$, which is a contradiction.

Similarly  for the case that $u_{k-1}$ and $u_{k+1}$ are the same,
$\lambda u_1\cdots u_{k-1}u_{k+2}\cdots u_n\lambda'$ is a path  of length
$\rho_{\mathcal H}(\lambda, \lambda')-2$ to connect vertices $\lambda$ and $\lambda'$. This is a contradiction.

\subsection{Proof of Proposition \ref{Vdistance.pr}}
\label{Vdistancepr.appendix}

The non-negativity and symmetry  is obvious, while
the identity of indiscernibles holds since there is no edge assigned in ${\mathcal H}$ between two distinct vertices in $V$.

Now we prove the triangle inequality
\begin{equation}\label{Vdistance.pr.pf.eq0}
\rho(i, i')\le \rho(i,i^{\prime\prime})+\rho(i^{\prime\prime}, i')
\ \ {\rm for\ distinct\  vertices} \ i, i', i^{\prime\prime}\in V.
\end{equation}
Let $m=\rho(i,i^{\prime\prime})$ and $n=\rho(i^{\prime\prime},i')$.
Take  a path $iv_1\ldots v_{m}i^{\prime\prime}$
of length $m+1$ to connect $i$ and $i^{\prime\prime}$, and
another path $i^{\prime\prime}u_1\ldots u_{n}i'$  of length $n+1$
to connect  $i^{\prime\prime}$ and $i'$.
If $v_{m}=u_1$, then
$iv_1\ldots v_{m}u_2\cdots u_{n}i'$ is a path of length $m+n$ to connect vertices $i$ and $i'$,
which implies that
\begin{equation}\label{Vdistance.pr.pf.eq1}
\rho(i,i')\le m+n-1<\rho(i,i^{\prime\prime})+\rho(i^{\prime\prime},i').
\end{equation}
If $v_{m}\ne u_1$, then $(v_{m}, u_1)$ is an edge in the graph ${\mathcal G}$  (and then also in the graph
${\mathcal H}$) by \eqref{basicassumption1}.
Thus $iv_1\ldots v_{m}u_1u_2\cdots u_{n}i'$ is a path of length $m+n+1$ to connect vertices $i$ and $i'$,
and
\begin{equation}\label{Vdistance.pr.pf.eq2}
\rho(i,i')\le m+n=\rho(i,i^{\prime\prime})+\rho(i^{\prime\prime}, i').
\end{equation}
Combining \eqref{Vdistance.pr.pf.eq1} and \eqref{Vdistance.pr.pf.eq2}
proves \eqref{Vdistance.pr.pf.eq0}.

\subsection{Proof of Proposition \ref{Vdoublingmeasure.pr}}
\label{Vdoublingmeasurepr.appendix}

To prove Proposition \ref{Vdoublingmeasure.pr}, we need  two  lemmas comparing measures of balls  in graphs ${\mathcal G}$ and ${\mathcal V}$.

\begin{Lm}\label{mumuG.lem1}
If ${\mathcal H}$ satisfies \eqref{basicassumption1} and \eqref{basicassumption2}, then
\begin{equation} \label{mumuG.lem.eq1}
\mu( B(i, r))\le  L \mu_{\mathcal G}(B_{\mathcal G}(\lambda, r))
\ \ {\rm for\ any} \ \lambda\in G\ {\rm  with} \ (i, \lambda)\in T.
\end{equation}
\end{Lm}

\begin{proof}
Let $i'\in B(i,r)$ with $i'\ne i$.
By Proposition \ref{shortestpath.pr},
there exists a path $\lambda_1\ldots \lambda_n$ of length $\rho(i,i')-1$
 in the graph ${\mathcal G}$ such that
$(i, \lambda_1), (i', \lambda_n)\in T$.
Then
\begin{equation*}
\rho_{\mathcal G}(\lambda,\lambda_n)\le
\rho_{\mathcal G}(\lambda,\lambda_1)+\rho_{\mathcal G}(\lambda_1,\lambda_n)
\le  \rho(i,i')\le r
\end{equation*}
as either $\lambda_1=\lambda$ or $(\lambda, \lambda_1)$ is an edge in ${\mathcal G}$ by   \eqref{basicassumption1}.
This shows that for any innovative position $i'\in B(i, r)$ there exists an anchor agent $\lambda_n$
in the ball $B_{\mathcal G}(\lambda, r)$. This observation  together with  \eqref{basicassumption2}
proves \eqref{mumuG.lem.eq1}.
\end{proof}

\begin{Lm}\label{mumuG.lem2}
If ${\mathcal H}$ satisfies \eqref{basicassumption0}, \eqref{basicassumption1} and \eqref{basicassumption3}, then
\begin{equation} \label{mumuG.lem2.eq1}
\mu_{\mathcal G}(B_{\mathcal G}(\lambda, r))\le  \Big(\sup_{\lambda'\in G} \mu_{\mathcal G}(B_{\mathcal G}(\lambda', 2M+2))\Big)
\mu(B(i, r+M+1))
\end{equation}
for any $\lambda\in G$ and $r\ge M+1$, where $(i, \lambda')\in T$ and $\lambda'\in B_{\mathcal G}(\lambda,  M)$.
\end{Lm}

\begin{proof} Let $\lambda_1=\lambda$ and take  $\Lambda=\{\lambda_m\}_{m\ge 1}$ 
such that
(i) $B_{\mathcal G}(\lambda_m, M+1)\subset B_{\mathcal G}(\lambda, r)$
 for all $\lambda_m\in \Lambda$;
(ii) $B_{\mathcal G}(\lambda_m, M+1)\bigcap B_{\mathcal G}(\lambda_{m'}, M+1)= \emptyset$
  for all  distinct vertices $\lambda_m, \lambda_{m'}\in \Lambda$; and
 (iii) $ B_{\mathcal G}(\tilde \lambda, M+1)\bigcap  \big(\bigcup_{\lambda_m\in \Lambda}B_{\mathcal G}(\lambda_m, M+1)\big) \neq \emptyset$ for all $\tilde \lambda\in B_{\mathcal G}(\lambda, r)$.
  The set $\Lambda$  could be considered as a maximal $(M+1)$-disjoint subset  of
  $B_{\mathcal G}(\lambda, r)$.
  Following the argument used in the proof of Proposition \ref{covering.pr},
  $ \{B_{\mathcal G}(\lambda_m, 2(M+1))\}_{\lambda_m\in \Lambda}$  forms a covering of the ball
  $B(\lambda, r)$, which implies that
    \begin{equation} \label{mumuG.lem2.pfeq4}
\hskip-0.1in\mu_{\mathcal G}(B_{\mathcal G}(\lambda, r))
 \le \Big(\sup_{\lambda_m\in \Lambda} \mu_{\mathcal G}(B_{\mathcal G}(\lambda_m, 2M+2))\Big) \# \Lambda \le
    \Big(\sup_{\lambda'\in G} \mu_{\mathcal G}(B_{\mathcal G}(\lambda', 2M+2))\Big) \# \Lambda.
\end{equation}

For $\lambda_m\in \Lambda$, define
$$V_{\lambda_m}=\{i'\in V, (i', \tilde \lambda)\in T \ {\rm for \ some} \ \tilde \lambda\in B_{\mathcal G}(\lambda_m, M)\}.$$
Then it follows from \eqref{basicassumption3} that
\begin{equation} \label{mumuG.lem2.pfeq5}
\# V_{\lambda_m}\ge 1 \ {\rm for \ all} \ \lambda_m\in \Lambda.
\end{equation}
Observe
that the distance of  anchor agents associated with innovative positions in
distinct $V_{\lambda_m}$  is at least 2 by the  second  requirement (ii) for the set $\Lambda$.
This together with the assumption
\eqref{basicassumption1}
implies that
\begin{equation}\label{mumuG.lem2.pfeq6}
V_{\lambda_m}\cap V_{\lambda_{m'}}=\emptyset\ \ {\rm  for\ distinct} \ \lambda_m, \lambda_{m'}\in \Lambda.
\end{equation}
Combining \eqref{mumuG.lem2.pfeq4}, \eqref{mumuG.lem2.pfeq5}
and \eqref{mumuG.lem2.pfeq6} leads to
\begin{equation} \label{mumuG.lem2.pfeq7}
\mu_{\mathcal G}(B_{\mathcal G}(\lambda, r))\le \Big(\sup_{\lambda'\in G} \mu_{\mathcal G}(B_{\mathcal G}(\lambda', 2M+2))\Big) \# \Big(\cup_{\lambda_m\in \Lambda} V_{\lambda_m}\Big).
\end{equation}

Take  $i\in V$ with $(i,\lambda')\in T$ for some $\lambda'\in B_{\mathcal G}(\lambda, M)$,
 and $i'\in V_{\lambda_m}, \lambda_m\in \Lambda$.
 Then
 $$\rho_{\mathcal H}(i, \lambda)\le \rho_{\mathcal H}(i, \lambda')+\rho_{\mathcal H}(\lambda', \lambda)\le M+1,$$
 and
 $$\rho_{\mathcal H}(i', \lambda)\le
 \rho_{\mathcal H}(i', \tilde \lambda)+\rho_{\mathcal H}(\tilde \lambda, \lambda)
 \le r+1,$$
 where $\tilde \lambda\in B_{\mathcal G}(\lambda_m, M)$ and $(i', \tilde \lambda)\in T$.
Thus
\begin{equation}\label{mumuG.lem2.pfeq8}
\rho(i,i')\le r+M+1.
\end{equation}
Then the desired estimate \eqref{mumuG.lem2.eq1} follows from
\eqref{mumuG.lem2.pfeq7} and \eqref{mumuG.lem2.pfeq8}.
 \end{proof}

We are ready to prove Proposition \ref{Vdoublingmeasure.pr}.
\begin{proof}[Proof of Proposition \ref{Vdoublingmeasure.pr}] First we prove the doubling property \eqref{Vdoublingmeasure.pr.eq1} for the measure $\mu$.
Take $i\in V$. Then for $r\ge 2(M+1)$ it follows from Lemmas \ref{mumuG.lem1}
and \ref{mumuG.lem2} that
\begin{eqnarray}\label{Vdoublingmeasure.pr.pf.eq1}
\hskip-0.08in \mu (B(i,2r)) &  \hskip-0.08in \le & \hskip-0.08in     L\mu_{\mathcal G}(B_{\mathcal G}(\lambda, 2r))
 \le
 L (D_0({\mathcal G}))^2 \mu_{\mathcal G}(B_{\mathcal G}(\lambda, r/2))\nonumber\\
  \hskip-0.08in &  \hskip-0.08in \le & \hskip-0.08in
 K L (D_0({\mathcal G}))^2
  \mu(B(i, r/2+M+1)) \le
 K L (D_0({\mathcal G}))^2
  \mu(B(i, r)),
\end{eqnarray}
where $\lambda\in G$ is a vertex with $(i, \lambda)\in T$ and
\begin{equation}\label{Vdoublingmeasure.pr.pf.eq2}K:=\sup_{\lambda'\in G} \mu_{\mathcal G}(B_{\mathcal G}(\lambda', 2M+2))\le \frac{((\deg ({\mathcal G}))^{2M+3}-1}{\deg({\mathcal G})-1}\end{equation}
by  \eqref{maximalvertex.def}.
From the doubling property \eqref{doublingconstants.def} for the measure $\mu_{\mathcal G}$,  we obtain
\begin{equation} \label{Vdoublingmeasure.pr.pf.eq3}
 \mu (B(i,2r))   \le K L
D_0({\mathcal G})  
\le  K L D_0({\mathcal G})
    \mu(B(i, r)) \  \ {\rm for} \ 0\le r\le 2(M+1).
  \end{equation}
  Then the doubling property
 \eqref{Vdoublingmeasure.pr.eq1} follows from \eqref{Vdoublingmeasure.pr.pf.eq1}, \eqref{Vdoublingmeasure.pr.pf.eq2} and \eqref{Vdoublingmeasure.pr.pf.eq3}.

Next we prove the doubling property \eqref{Vdoublingmeasure.pr.eq2}
for the measure $\mu_{\mathcal G}$.
Let $\lambda'\in B_{\mathcal G}(\lambda, M)$ with $(i, \lambda')\in T$ for some $i\in V$.
The existence of such $\lambda'$ follows from assumption
\eqref{basicassumption3}.
From Lemmas \ref{mumuG.lem1}
and \ref{mumuG.lem2},  we obtain
\begin{eqnarray}\label{Vdoublingmeasure.pr.pf.eq4}
\hskip-0.08in \mu_{\mathcal G} (B_{\mathcal G}(\lambda,2r)) &  \hskip-0.08in \le & \hskip-0.08in
  K  \mu(B(i, 2r+M+1))
 \le
  D_0^2 K
  \mu\Big(B\Big(i, \frac{r}2+\frac{(M+1)}4\Big)\Big)
  \nonumber\\
  \hskip-0.08in &  \hskip-0.08in \le & \hskip-0.08in
  D_0^2 L K
  \mu_{\mathcal G} \Big(B_{\mathcal G}\Big(\lambda', \frac{r}{2}+\frac{M+1}{4}\Big)\Big)
  \nonumber\\
\hskip-0.08in &  \hskip-0.08in \le & \hskip-0.08in
  D_0^2 L K  \mu_{\mathcal G} \Big(B_{\mathcal G}\Big(\lambda, \frac{r}{2}+\frac{M+1}{4}+M\Big)\Big)\le
  D_0^2 L K
  \mu_{\mathcal G} (B_{\mathcal G}(\lambda, r))
\end{eqnarray}
for $r\ge 3M$, and
\begin{eqnarray}\label{Vdoublingmeasure.pr.pf.eq5}
 \hskip-0.08in  \mu_{\mathcal G} (B_{\mathcal G}(\lambda,2r))
 &  \hskip-0.08in \le  & \hskip-0.08in
  K  \mu(B(i, 7M))
 \le D_0^2 K
  \mu(B(i, 2M))
 \nonumber\\
\hskip-0.08in &  \hskip-0.08in \le & \hskip-0.08in
  D_0^2 L K
  \mu_{\mathcal G} (B_{\mathcal G}(\lambda', 2M))
 \le
 D_0^2 L K^2
 \mu_{\mathcal G} (B_{\mathcal G}(\lambda, r))
\end{eqnarray}
for $0\le r\le 3M-1$.
Combining \eqref{Vdoublingmeasure.pr.pf.eq2}, \eqref{Vdoublingmeasure.pr.pf.eq4} and
\eqref{Vdoublingmeasure.pr.pf.eq5}
proves \eqref{Vdoublingmeasure.pr.eq2}.\end{proof}

\subsection{Proof of Proposition \ref{Vpolynomial.pr}}
\label{Vpolynomialpr.appendix}

The polynomial growth property \eqref{Vpolynomial.pr.eq1} for the measure $\mu$
follows immediately from
 Lemma \ref{mumuG.lem1}.

The polynomial growth property \eqref{Vpolynomial.pr.eq2} for the measure $\mu_{\mathcal G}$
holds because
\begin{equation*}\label{Vpolynomial.pr.pfeq1}
\mu_{\mathcal G}(B_{\mathcal G}(\lambda, r))
\le \frac{(\deg ({\mathcal G}))^{M}-1}{\deg({\mathcal G})-1}, \ 0\le r\le M-1
\end{equation*}
by  \eqref{maximalvertex.def},
and
\begin{eqnarray*}
\mu_{\mathcal G}(B_{\mathcal G}(\lambda, r))
& \hskip-0.08in \le & \hskip-0.08in
D_1 \Big(\frac{(\deg ({\mathcal G}))^{2M+3}-1}{\deg({\mathcal G})-1}\Big) (r+M+2)^d\nonumber \\
& \le &  2^d D_1 \Big(\frac{(\deg ({\mathcal G}))^{2M+3}-1}{\deg({\mathcal G})-1}\Big) (r+1)^d, \ r\ge M,
\end{eqnarray*}
 by
 \eqref{Vdoublingmeasure.pr.pf.eq2} and Lemma \ref{mumuG.lem2}.

\subsection{Proof of Proposition \ref{jaffard.pr} }
\label{jaffardpr.appendix}

To prove Proposition \ref{jaffard.pr}, we need a  lemma.

\begin{Lm} \label{jaffardpr.lem}
Let ${\mathcal G}$ be a connected simple graph. If
its counting measure has  polynomial growth \eqref{countmeasure.pr.eq1}, then
\begin{equation} \label{jaffardpr.pf.eq2}
\sup_{\lambda\in G} \sum_{\rho_{\mathcal G}(\lambda, \lambda')\ge s} (1+\rho_{\small {\mathcal G}}(\lambda,\lambda'))^{-\alpha}
\le \frac{D_1({\mathcal G})\alpha}{\alpha-d} (s+1)^{-\alpha+d}
\end{equation}
for all $\alpha>d$ and nonnegative integers $s$, where $d$ and $D_1({\mathcal G})$ are  the Beurling dimension and  sampling density respectively. 
\end{Lm}

\begin{proof} Take $\lambda\in G$ and  $\alpha>d$. Then
 \begin{eqnarray}  \hskip-0.08in &\hskip-0.08in  & \hskip-0.08in
  \sum_{\rho_{\mathcal G}(\lambda, \lambda')\ge s} (1+\rho_{\small {\mathcal G}}(\lambda,\lambda'))^{-\alpha}
 =  \sum_{n\ge s} (n+1)^{-\alpha} \Big(\sum_{\rho_{\small {\mathcal G}}(\lambda,\lambda')=n} 1\Big)\nonumber\\
\hskip-0.08in &\hskip-0.08in \le & \hskip-0.08in
  \sum_{n\ge s} \mu_{\mathcal G}(B_{\mathcal G}(\lambda,n)) ((n+1)^{-\alpha}-(n+2)^{-\alpha}\big)\nonumber\\
\hskip-0.08in &\hskip-0.08in \le & \hskip-0.08in  D_1({\mathcal G})
 \sum_{n=s}^\infty (n+1)^{d} ((n+1)^{-\alpha}-(n+2)^{-\alpha}\big)
\nonumber\\
\hskip-0.08in &\hskip-0.08in = &\hskip-0.08in  D_1({\mathcal G})
\Big((s+1)^{-\alpha+d} + \sum_{n=s+1}^\infty  (n+1)^{-\alpha} \big( (n+1)^{d}-n^{d}\big) \Big)
\nonumber\\
 \hskip-0.08in &\hskip-0.08in \le &\hskip-0.08in  D_1({\mathcal G})
\Big((s+1)^{-\alpha+d}+  d\! \int_{s+1}^\infty t^{d-\alpha-1} dt\Big)  =  \frac{D_1({\mathcal G})\alpha}{\alpha-d} (s+1)^{-\alpha+d}, \end{eqnarray}
where the second inequality follows from  \eqref{countmeasure.pr.eq1},
and the third one is true as $(n+1)^{d}-n^{d}\le d (n+1)^{d-1}$ for $n\ge 1$ and $d\ge 1$.
\end{proof}

 Now we prove Proposition \ref{jaffard.pr}.

 \begin{proof}[Proof of Proposition \ref{jaffard.pr}]
 Take ${\mathbf A} 
 \in {\mathcal J}_\alpha({\mathcal G}, {\mathcal V})$ and ${\mathbf c}:=(c(i))_{i\in V}\in \ell^p, 1<p<\infty$. 
 Then
   \begin{eqnarray} \label{jaffardpr.pf.eq5} 
 \|{\mathbf A}{\mathbf c}\|_p^p& \hskip-0.08in\le & \hskip-0.08in
   \|{\mathbf A}\|_{{\mathcal J}_\alpha({\mathcal G}, {\mathcal V})}^p
 \sum_{\lambda\in G} \Big(\sum_{i\in V} (1+\rho_{\mathcal H}(\lambda, i))^{-\alpha} |c(i)|\Big)^p
 \nonumber\\
 &  \hskip-0.08in\le &  \hskip-0.08in
  \|{\mathbf A}\|_{{\mathcal J}_\alpha({\mathcal G}, {\mathcal V})}^p
 \|{\mathbf c}\|_p^p
 \Big( \sup_{\lambda'\in G} \sum_{i'\in V} (1+\rho_{\mathcal H}(\lambda',i'))^{-\alpha}\Big)^{p-1}
 \Big( \sup_{i'\in V} \sum_{\lambda'\in G} (1+\rho_{\mathcal H}(\lambda',i'))^{-\alpha}\Big).
 \end{eqnarray}
For any $\lambda'\in G$ and $i'\in V$, it follows from Proposition \ref{shortestpath.pr} that
 \begin{equation} \label{jaffardpr.pf.eq5a}
  \rho_{\mathcal G}(\lambda', \lambda^{\prime\prime})+1\ge \rho_{\mathcal H}(\lambda', i')\ge \rho_{\mathcal G}(\lambda', \lambda^{\prime\prime})\ \
    {\rm for\ all} \ \lambda^{\prime\prime}\in G\ {\rm  with} \ (i', \lambda^{\prime\prime})\in T. \end{equation}
By  \eqref{basicassumption2},  \eqref{Vdim.cr.eq1}, \eqref{jaffardpr.pf.eq5a} and Lemma
\ref{jaffardpr.lem},
  we  obtain
 \begin{eqnarray}\label{jaffardpr.pf.eq5b}
\sum_{i'\in V} (1+\rho_{\mathcal H}(\lambda',i'))^{-\alpha}
 \hskip-0.03in &\hskip-0.08in \le & \hskip-0.08in \sum_{\lambda^{\prime\prime}\in G} \Big(\sum_{(i', \lambda^{\prime\prime})\in T} 1\Big)
  (1+\rho_{\mathcal G}(\lambda',\lambda^{\prime\prime}))^{-\alpha}\nonumber\\
 \hskip-0.08in &\hskip-0.08in \le  & \hskip-0.08in L \sum_{\lambda^{\prime\prime}\in G} (1+\rho_{\mathcal G}(\lambda',\lambda^{\prime\prime}))^{-\alpha}
 \le \frac{L D_1({\mathcal G})\alpha }{\alpha-d}\ {\rm \ for\  any}\  \lambda'\in G,
 \end{eqnarray}
and
 \begin{equation}\label{jaffardpr.pf.eq5c}
 \sum_{\lambda'\in G} (1+\rho_{\mathcal H}(\lambda',i'))^{-\alpha} \le  \sum_{\lambda'\in G}  (1+\rho_{\mathcal G}(\lambda',\lambda^{\prime\prime}))^{-\alpha}
 \le \frac{D_1({\mathcal G})\alpha }{\alpha-d} \ {\rm \ for\  any}\  i'\in V,
 \end{equation}
  where $\lambda^{\prime\prime}\in G$ satisfies $(i', \lambda^{\prime\prime})\in T$.
Combining  \eqref{jaffardpr.pf.eq5}, \eqref{jaffardpr.pf.eq5b}
and \eqref{jaffardpr.pf.eq5c}  proves \eqref{jaffardpr.eq1} for $1<p<\infty$.

We can use similar argument to prove \eqref{jaffardpr.eq1} for $p=1, \infty$.
\end{proof}

\subsection{Proof of Proposition \ref{bandapproximation.pr} }
\label{bandapproximationpr.appendix}

Following the proof of Proposition \ref{jaffard.pr}, we obtain
\begin{eqnarray}\label{bandapproximation.pr.pf.eq2}
   \|({\mathbf A}-{\mathbf A}_s) {\mathbf c}\|_p
  &\hskip-0.08in \le & \hskip-0.08in  \|{\mathbf A}\|_{{\mathcal J}_\alpha({\mathcal G}, {\mathcal V})}
 \Big( \sup_{\lambda'\in G} \sum_{\rho_{\mathcal H}(\lambda', i')>s} (1+\rho_{\mathcal H}(\lambda',i'))^{-\alpha}\Big)^{1-1/p}
 \nonumber\\
\hskip-0.18in & \hskip-0.08in & \hskip-0.08in\times
 \Big( \sup_{i'\in V} \sum_{\rho_{\mathcal H}(\lambda', i')>s} (1+\rho_{\mathcal H}(\lambda',i'))^{-\alpha}\Big)^{1/p}
   \|{\mathbf c}\|_p,
  \end{eqnarray}
  where ${\mathbf c}\in \ell^p, 1\le p\le \infty$.
Applying similar argument used to prove \eqref{jaffardpr.pf.eq2},
\eqref{jaffardpr.pf.eq5b} and \eqref{jaffardpr.pf.eq5c}, we have
\begin{equation}\label{bandapproximation.pr.pf.eq1}
\sup_{\lambda'\in G} \sum_{\rho_{\mathcal H}(\lambda', i')>s} (1+\rho_{\mathcal H}(\lambda',i'))^{-\alpha}
\le  L \sup_{\lambda'\in G}
\sum_{\rho_{\mathcal G}(\lambda', \lambda'')\ge s} (1+\rho_{\mathcal G}(\lambda',\lambda''))^{-\alpha} \le
 \frac{ D_1({\mathcal G}) L\alpha}{\alpha-d} (s+1)^{-\alpha+d}
\end{equation}
and
\begin{equation}\label{bandapproximation.pr.pf.eq3}
 \sup_{i'\in V} \sum_{\rho_{\mathcal H}(\lambda', i')>s} (1+\rho_{\mathcal H}(\lambda',i'))^{-\alpha} 
\le \frac{ D_1({\mathcal G}) \alpha}{\alpha-d} (s+1)^{-\alpha+d}.
\end{equation}
Then the approximation error estimate \eqref{bandapproximation.pr.eq1}
follows from \eqref{bandapproximation.pr.pf.eq2},
\eqref{bandapproximation.pr.pf.eq1} and \eqref{bandapproximation.pr.pf.eq3}.

\subsection{Proof of Theorem \ref{wienerlemma.tm}}
\label{wiener.appendix}

To prove Wiener's lemma (Theorem  \ref{wienerlemma.tm}) for ${\mathcal J}_\alpha({\mathcal V}), \alpha>d$,
 we first show that
 it  is a Banach algebra of matrices.

\begin{pr}\label{jaffard.banach.pr}
 Let ${\mathcal V}$ be an undirected graph with the  counting measure $\mu$
 having
 polynomial growth \eqref{mupolynomial.def}.
 Then for any $\alpha>d$, ${\mathcal J}_\alpha({\mathcal V})$ is a  Banach algebra of matrices:
\begin{itemize}
\item[{(i)}] $\|\beta {\bf C}\|_{\mathcal{J}_\alpha({\mathcal V})} = |\beta| \|{\bf C}\|_{\mathcal{J}_{\alpha}({\mathcal V})}$;
\item [{(ii)}]
$\|{\bf C}+{\bf D}\|_{\mathcal{J}_{\alpha}({\mathcal V})}\le \|{\bf C}\|_{\mathcal{J}_{\alpha}({\mathcal V})}+  \|{\bf D}\|_{\mathcal{J}_{\alpha}({\mathcal V})}$;
\item [{(iii)}]  $\|{\bf C} {\bf D}\|_{\mathcal{J}_{\alpha}({\mathcal V})}\le
\frac{2^{\alpha+1}D_1\alpha}{\alpha-d} \|{\bf C}\|_{\mathcal{J}_{\alpha}({\mathcal V})} \|{\bf D}\|_{\mathcal{J}_{\alpha}({\mathcal V})}$; and
\item[{(iv)}] $\|{\bf D}{\bf c}\|_2\le \frac{D_1\alpha}{\alpha-d} \|{\bf D}\|_{\mathcal{J}_{\alpha}({\mathcal V})}\|{\bf c}\|_2$
\end{itemize}
for any scalar $\beta$, vector ${\bf c}\in \ell^2$ and matrices ${\bf C}, {\bf D}\in \mathcal{J}_{\alpha}({\mathcal V})$.
\end{pr}

\begin{proof} The first two conclusions follow immediately from
\eqref{Vjaffard.def1} and \eqref{Vjaffard.def2}. 

Now we prove the third conclusion.  Take ${\bf C}, {\bf D}\in {\mathcal J}_\alpha({\mathcal V})$.
Then
\begin{eqnarray}\label{jaffard.banach.pr.pfeq1}
\|{\mathbf C}{\mathbf D}\|_{{\mathcal J}_\alpha({\mathcal V})}
&  \hskip-0.1in \le & \hskip-0.1in
2^\alpha \|{\mathbf C}\|_{{\mathcal J}_\alpha({\mathcal V})}
\|{\mathbf D}\|_{{\mathcal J}_\alpha({\mathcal V})}   \sup_{i,i'\in V}
\Big(\sum_{\rho(i,i^{\prime\prime})\ge \rho(i,i')/2}  (1+\rho(i^{\prime\prime}, i'))^{-\alpha}
\nonumber\\
& & +
\sum_{\rho(i^{\prime\prime}, i')\ge \rho(i,i')/2}
(1+\rho(i,i^{\prime\prime}))^{-\alpha} \Big).
\end{eqnarray}
Following the argument used in the proofs of Lemma \ref{jaffardpr.lem},
 we have
\begin{equation} \label{jaffard.banach.pr.pfeq2}
\sup_{i\in V} \sum_{\rho(i,i')\ge s} (1+\rho(i, i'))^{-\alpha}
\le \frac{D_1\alpha}
{\alpha-d} (s+1)^{-\alpha+d}, \ 0\le s\in \ZZ.
\end{equation}
Combining \eqref{jaffard.banach.pr.pfeq1} and \eqref{jaffard.banach.pr.pfeq2}
 proves the third conclusion.

 Following the proof of Proposition \ref{jaffard.pr}
 and applying \eqref{jaffard.banach.pr.pfeq2} instead of
\eqref{jaffardpr.pf.eq5b} and \eqref{jaffardpr.pf.eq5c}, we obtain the fourth conclusion.
\end{proof}

Now, we prove Theorem \ref{wienerlemma.tm}.

\begin{proof}[Proof of Theorem \ref{wienerlemma.tm}]
 Following the argument in \cite{suncasp05}, it suffices to establish the following differential norm
inequality:
\begin{equation}\label{wienerlemma.pf.eq1}
\|{\mathbf C}^2\|_{{\mathcal J}_\alpha({\mathcal V})}
 \le  2^{\alpha+d/2+2}
D_1^{1/2} (D_1\alpha/(\alpha-d))^{1-\theta} (\|{\mathbf C}\|_{{\mathcal J}_\alpha({\mathcal V})})^{2-\theta} (\|{\mathbf C}\|_{{\mathcal B}^2})^{\theta}
 \end{equation}
 holds for all ${\mathbf C}\in {\mathcal J}_\alpha({\mathcal V})$,
 where $\theta=(2\alpha-2d)/(2\alpha-d)\in (0,1)$.

Write ${\bf C}=(c(i,i'))_{i,i'\in V}$. Then
\begin{eqnarray}\label{wienerlemma.pf.eq2}
\|{\mathbf C}^2\|_{{\mathcal J}_\alpha({\mathcal V})} 
 & \hskip-0.1in \le & \hskip-0.1in 2^\alpha \|{\mathbf C}\|_{{\mathcal J}_\alpha(\mathcal V)}  \Big(\sup_{i,i'\in V} \sum_{\rho(i,i^{\prime\prime})\ge \rho(i,i')/2}  |c(i^{\prime\prime},i')| +\sup_{i,i'\in V} \sum_{\rho(i^{\prime\prime}, i')\ge \rho(i,i')/2}  |c(i,i^{\prime\prime})|\Big)\nonumber\\
 & \hskip-0.1in \le & \hskip-0.1in 2^\alpha \|{\mathbf C}\|_{{\mathcal J}_\alpha(\mathcal V)}  \Big(\sup_{i'\in V} \sum_{i^{\prime\prime}\in V}  |c(i^{\prime\prime},i')|  +\sup_{i\in V} \sum_{i^{\prime\prime}\in V}  |c(i,i^{\prime\prime})|\Big).
\end{eqnarray}
Set
\begin{equation}\label{tau.def}
\tau := \Big(\frac{D_1\alpha
  \|{\bf C}\|_{\mathcal{J}_{\alpha}({\mathcal V})}}
  {(\alpha-d)\|{\bf C}\|_{{\mathcal B}^2}}\Big)^{2/(2\alpha-d)}\ge 1\end{equation}
  by Proposition \ref{jaffard.banach.pr}.
      For $i'\in V$, we obtain
\begin{eqnarray}\label {wienerlemma.pf.eq3}
  \sum_{i^{\prime\prime}\in V}  |c(i^{\prime\prime},i')|
& \hskip-0.1in \le & \hskip-0.1in \Big(\sum_{\rho(i^{\prime\prime},i')\le
\tau}|c(i^{\prime\prime},i')|^2\Big)^{1/2} \Big(\sum_{\rho(i^{\prime\prime},i')\le \tau}1\Big)^{1/2} + \|{\mathbf C}\|_{{\mathcal J}_\alpha({\mathcal V})}
 \sum_{\rho(i^{\prime\prime},i')>\tau} (1+\rho(i^{\prime\prime},i'))^{-\alpha}\nonumber\\
& \hskip-0.1in \le & \hskip-0.1in D_1^{1/2} \|{\mathbf C}\|_{{\mathcal B}^2} (1+\lfloor\tau\rfloor)^{d/2} +D_1\alpha (\alpha-d)^{-1}\|{\mathbf C}\|_{{\mathcal J}_\alpha(\mathcal V)} (1+\lfloor\tau\rfloor)^{-\alpha+d}\nonumber\\
& \hskip-0.1in \le & \hskip-0.1in 2^{d/2+1}D_1^{1/2} (D_1\alpha/(\alpha-d))^{d/(2\alpha-d)} (\|{\mathbf C}\|_{{\mathcal J}_\alpha(\mathcal V)})^{1-\theta} (\|{\mathbf C}\|_{{\mathcal B}^2})^{\theta},
\end{eqnarray}
where
the second inequality holds by \eqref{jaffard.banach.pr.pfeq2}
and the last inequality follows from \eqref{tau.def}.  Similarly, for $i\in V$ we have
  \begin{equation}\label{wienerlemma.pf.eq4}
  \sum_{i^{\prime\prime}\in V}  |c(i',i^{\prime\prime})|
\le  2^{d/2+1}D_1^{1/2} (D_1\alpha/(\alpha-d))^{d/(2\alpha-d)} (\|{\mathbf C}\|_{{\mathcal J}_\alpha(\mathcal{V})})^{1-\theta} (\|{\mathbf C}\|_{{\mathcal B}^2})^{\theta}.
\end{equation}
Combining \eqref{wienerlemma.pf.eq2}, \eqref{wienerlemma.pf.eq3}
and \eqref{wienerlemma.pf.eq4} proves \eqref{wienerlemma.pf.eq1}. This completes the proof of Theorem \ref{wienerlemma.tm}.
\end{proof}

\subsection{Proof of Theorem \ref{lpstability.tm}}
\label{lpstabilitytm.appendix}

To prove Theorem \ref{lpstability.tm}, we need Theorem \ref{wienerlemma.tm} and the following lemma about families ${\mathcal J}_\alpha({\mathcal G}, {\mathcal V})$ and ${\mathcal J}_\alpha({\mathcal V})$ of matrices.

\begin{Lm}\label{jaffard.algebra2.lm}
Let  ${\mathcal G}, {\mathcal H}, {\mathcal V}$ and d be as in Proposition \ref{linftystability.pr}.
Then
\begin{itemize}
\item[{(i)}]
$\|{\mathbf A}{\mathbf C}\|_{{\mathcal J}_\alpha({\mathcal G}, {\mathcal V})}
\le \frac{2^{\alpha+1} LD_1({\mathcal G})\alpha}{\alpha-d}\|{\mathbf A}\|_{{\mathcal J}_\alpha({\mathcal G}, {\mathcal V})}
\|{\mathbf C}\|_{{\mathcal J}_\alpha({\mathcal V})}$ for all ${\mathbf A}\in {\mathcal J}_\alpha({\mathcal G}, {\mathcal V})$ and ${\mathbf C}\in {\mathcal J}_\alpha({\mathcal V})$.
\item[{(ii)}]
$\|{\mathbf A}^T{\mathbf B}\|_{{\mathcal J}_\alpha({\mathcal V})}
\le \frac{2^{\alpha+1} D_1({\mathcal G}) \alpha}{\alpha-d}\|{\mathbf A}\|_{{\mathcal J}_\alpha({\mathcal G}, {\mathcal V})}
\|{\mathbf B}\|_{{\mathcal J}_\alpha({\mathcal G}, {\mathcal V})}$
for all ${\mathbf A}, {\mathbf B}\in {\mathcal J}_\alpha({\mathcal G}, {\mathcal V})$.
\end{itemize}
\end{Lm}

\begin{proof}
Take
${\mathbf A}\in {\mathcal J}_\alpha({\mathcal G}, {\mathcal V})$ and ${\mathbf C}\in {\mathcal J}_\alpha({\mathcal V})$.
Observe from \eqref{basicassumption1} that
$$\rho_{\mathcal H}(\lambda, i)\le \rho_{\mathcal H}(\lambda, i')+\rho(i', i)
\ {\rm for \ all} \ \lambda\in G \ {\rm and} \ i, i'\in V.$$
Similar to the argument used in the proof of  Proposition \ref{jaffard.banach.pr}, we obtain
\begin{equation*}
\|{\mathbf A}{\mathbf C}\|_{{\mathcal J}_\alpha({\mathcal G}, {\mathcal V})}
\le  2^\alpha \|{\mathbf A}\|_{{\mathcal J}_\alpha({\mathcal G}, {\mathcal V})}
\|{\mathbf C}\|_{{\mathcal J}_\alpha({\mathcal V})} \Big(\sup_{i\in V}
\sum_{ i'\in V}
(1+\rho(i', i))^{-\alpha}
+\sup_{\lambda\in G}
\sum_{i'\in V}
 (1+\rho_{\mathcal H}(\lambda, i'))^{-\alpha} \Big). 
\end{equation*}
This together with \eqref{Vdim.cr.eq2},
\eqref{bandapproximation.pr.pf.eq1} and \eqref{jaffard.banach.pr.pfeq2} proves the first conclusion.

Recall that
\begin{equation}\label{jaffard.algebra2.pr.pfeq1}
\rho(i, i')\le \rho_{\mathcal H}(\lambda, i)+\rho_{\mathcal H}(\lambda, i')\ {\rm for \ all} \ \lambda\in G\ {\rm and} \ i, i'\in V.\end{equation}
Then for
${\mathbf A}, {\mathbf B}\in {\mathcal J}_\alpha({\mathcal G}, {\mathcal V})$, we obtain from
\eqref{bandapproximation.pr.pf.eq3} and \eqref{jaffard.algebra2.pr.pfeq1} that
\begin{eqnarray*}
\|{\mathbf A}^T{\mathbf B}\|_{{\mathcal J}_\alpha({\mathcal V})}
 &\hskip-0.08in  \le & \hskip-0.08in  2^{\alpha+1} \|{\mathbf A}\|_{{\mathcal J}_\alpha({\mathcal G}, {\mathcal V})}
\|{\mathbf B}\|_{{\mathcal J}_\alpha({\mathcal G}, {\mathcal V})}
\sup_{i\in V} \sum_{\lambda\in G} (1+\rho_{\mathcal H}(\lambda, i))^{-\alpha}\nonumber\\
&\hskip-0.08in \le &\hskip-0.08in  \frac{2^{\alpha+1} D_1({\mathcal G}) \alpha}{\alpha-d} \|{\mathbf A}\|_{{\mathcal J}_\alpha({\mathcal G}, {\mathcal V})}
\|{\mathbf B}\|_{{\mathcal J}_\alpha({\mathcal G}, {\mathcal V})}.
\end{eqnarray*}
This  completes the proof of the second conclusion.
\end{proof}

Now we prove Theorem \ref{lpstability.tm}.
\begin{proof}[Proof of Theorem \ref{lpstability.tm}]
Take ${\mathbf A}\in {\mathcal J}_\alpha({\mathcal G}, {\mathcal V})$ that has $\ell^2$-stability.
 Then  ${\mathbf A}^T{\mathbf A}$ has bounded inverse on $\ell^2$.
 Observe that ${\mathbf A}^T {\mathbf A} \in {\mathcal J}_\alpha({\mathcal V})$ by Lemma \ref{jaffard.algebra2.lm}.
 Therefore
 $({\mathbf A}^T {\mathbf A})^{-1}\in {\mathcal J}_{\alpha}({\mathcal V})$
 and
 ${\mathbf A}({\mathbf A}^T {\mathbf A})^{-1}\in {\mathcal J}_\alpha({\mathcal G}, {\mathcal V})$ by
  Theorem \ref{wienerlemma.tm}  and Lemma \ref{jaffard.algebra2.lm}.
  Hence for any ${\bf c}\in \ell^p$,
\begin{equation*}
\|{\bf c}\|_p  = \| ({\mathbf A}^T {\mathbf A})^{-1} {\mathbf A}^T {\mathbf A} {\bf c}\|_p \le
    \frac{D_1({\mathcal G})L \alpha}{\alpha-d}
 \|  {\mathbf A} ({\mathbf A}^T {\mathbf A})^{-1} \|_{{\mathcal J}_\alpha({\mathcal G}, {\mathcal V})}
\|{\mathbf A} {\bf c}\|_p
\end{equation*}
and
\begin{equation*}
\|{\mathbf A} {\bf c}\|_p \le \frac{D_1({\mathcal G})L \alpha}{\alpha-d}
 \|  {\mathbf A} \|_{{\mathcal J}_\alpha({\mathcal G}, {\mathcal V})}
\|{\bf c}\|_p
\end{equation*}
by Proposition \ref{jaffard.pr} and the dual property between sequences $\ell^p$ and $\ell^{p/(p-1)}$. The $\ell^p$-stability for the matrix ${\mathbf A}$ then follows.
\end{proof}

\subsection{Proof of Theorem  \ref{leastsqaure.tm}}
\label{leastsquaretm.appendix}
 The conclusion \eqref{l2estimate.in.eq} follows immediately from Proposition \ref{jaffard.pr}, Theorem \ref{wienerlemma.tm} and Lemma \ref{jaffard.algebra2.lm}.

\subsection{Proof of Theorem \ref{l2stablility.thm1}}
\label{l2stability.appendix}

 Observe from Proposition \ref{shortestpath.pr} that
 $$ B_{\mathcal H}(\gamma, r)\cap G=\{\gamma'\in G,\ \rho_{\mathcal G}(\gamma, \gamma')\le r\}, \ \gamma\in G.
 $$
 and
  $$ B_{\mathcal H}(i, r)\cap V=\{i'\in V,\ \rho(i, i')\le \max(r-1, 0)\}, \ i\in V.
 $$

Take ${\bf c}=(c(i))_{i\in V}$ supported in $B_{\mathcal H}(\lambda, N)\cap V$ and write ${\bf A}{\bf c}=(d(\lambda'))_{\lambda'\in G}$.
Then
\begin{equation}\label{l2stablility.thm1.pf.eq1}
\|{\bf A} {\bf c}\|_2\ge A \|{\mathbf A}\|_{{\mathcal J}_\alpha({\mathcal G}, {\mathcal V})}\|{\bf c}\|_2
\end{equation}
and
\begin{eqnarray}\label{l2stablility.thm1.pf.eq2}
 \sum_{\rho_{\mathcal H}(\lambda', \lambda)>2N} |d(\lambda')|^2
& \hskip-0.1in \le & \hskip-0.1in
 L D_1({\mathcal G}) N^{-\alpha+d} \|{\mathbf A}\|_{{\mathcal J}_\alpha({\mathcal G}, {\mathcal V})}^2
 \nonumber\\
& & \hskip-0.1in\times
 \sum_{\rho_{\mathcal H}(\lambda', \lambda)>2N}
 \sum_{i\in B_{\mathcal H}(\lambda,N)\cap V} (1+\rho_{\mathcal H}(\lambda',i))^{-\alpha}
|c(i)|^2
\nonumber\\
& \hskip-0.1in \le & \hskip-0.1in \big( D_1({\mathcal G})\big)^2 L N^{-2\alpha+2d} \alpha (\alpha-d)^{-1}
 \|{\mathbf A}\|_{{\mathcal J}_\alpha({\mathcal G}, {\mathcal V})}^2 \|{\bf c}\|_2^2,
\end{eqnarray}
where the first inequality holds as
$$\rho_{\mathcal H}(\lambda',i')\ge \rho_{\mathcal H}(\lambda',\lambda)-\rho_{\mathcal H}(i', \lambda)>N$$
for all $\lambda'\not\in B_{\mathcal H}(\lambda, 2N)$ and $i'\in B_{\mathcal H}(\lambda,N)$, and the last inequality follows from
\eqref{bandapproximation.pr.pf.eq3}.
Combining \eqref{l2stablility.thm1.pf.eq1} and \eqref{l2stablility.thm1.pf.eq2}
proves \eqref{l2stablility.thm1.eq1}.

\subsection{Proof of Theorem \ref{l2stablility.thm}}
In this subsection, we will prove the following strong version of Theorem \ref{l2stablility.thm}.

\begin{Tm}\label{l2stablility.newthm}
 Let ${\mathcal G}, {\mathcal H}, {\mathcal V}$  and
 ${\mathbf A}$ be as in Theorem \ref{l2stablility.thm}.
If there exists a positive constant $A_0$, an integer $N_0\ge 3$, and
a  maximal $\frac{N_0}4$-disjoint subset $G_{N_0/4}$
 such that
\eqref{l2stablility.thm2.eq0}
 is true  and
 \eqref{l2stablility.thm2.eq1}
hold for all  $\lambda_m\in G_{N_0/4}$,
 then  ${\mathbf A}$ satisfies 
  \eqref{l2stablility.thm2.eq2}.
\end{Tm}


\begin{proof}  Let $\psi_{0}$ be  the trapezoid function,
\begin{equation}\label{psi0.def}
\psi_0(t)=\left\{\begin{array}{ll}
1 & {\rm if} \ |t|\le 1/2\\
2-2|t| & {\rm if} \ 1/2<|t|\le 1\\
0 & {\rm if} \ |t|> 1.
\end{array}
\right.
\end{equation}
For $\lambda\in G$,  define multiplication operators $\Psi_{\lambda, V}^N$ and $\Psi_{\lambda, G}^N$ by
\begin{equation}\label{Psi.operator}
\Psi_{\lambda, V}^N:\  (c(i))_{i\in V} \longmapsto \big(\psi_{0}(\rho_{\mathcal H}(\lambda ,i)/N)c(i)\big)_{i\in V},
\end{equation}
\begin{equation}\label{Psi1.operator}
\Psi_{\lambda, G}^N:\  (d(\lambda'))_{\lambda'\in G} \longmapsto \big(\psi_{0}(\rho_{\mathcal H}(\lambda ,\lambda')/N)d(\lambda')\big)_{\lambda'\in G}.
\end{equation}
Observe that
$${\mathbf A}_{N}\Psi_{\lambda, V}^N={\mathbf A}_{N}\chi_{\lambda,V}^{N} \Psi_{\lambda, V}^N =\chi_{\lambda,G}^{2N}{\mathbf A}_{N}\chi_{\lambda,V}^{N} \Psi_{\lambda, V}^N,
N\ge 0,$$
where ${\mathbf A}_{N}$ is a band approximation of the matrix ${\bf A}$ in \eqref{bandapproximation.def}.
Then  for  all $\lambda_m\in G_{N_0/4}$,
it follows from Proposition \ref{bandapproximation.pr} and our local stability assumption  \eqref{l2stablility.thm2.eq1} that
\begin{eqnarray*} 
  \|{\mathbf A}_{N_0} \Psi_{\lambda_m,V}^{N_0} {\bf c}\|_2
&\hskip-0.1in \ge & \hskip-0.1in \|\chi_{\lambda_m,G}^{2N_0} {\mathbf A}\chi_{\lambda_m,V}^{N_0} \Psi_{\lambda_m,V}^{N_0} {\bf c}\|_2 -
\|\chi_{\lambda_m,G}^{2N_0}({\mathbf A}-{\mathbf A}_{N_0}) \Psi_{\lambda_m,V}^{N_0} {\bf c}\|_2\nonumber\\
& \hskip-0.1in \ge & \hskip-0.1in \Big(A_0-\frac{D_1({\mathcal G}) L \alpha}{\alpha-d} N_0^{-\alpha+d}\Big) \|{\bf A}\|_{{\mathcal J}_\alpha({\mathcal G}, {\mathcal V})} \|\Psi_{\lambda_m,V}^{N_0} {\bf c}\|_2, \ \  {\bf c}\in \ell^2.
\end{eqnarray*}
Therefore
\begin{eqnarray} \label{l2stablility.thm.pf.eq4}
&& \Big(\sum_{\lambda_m\in G_{N_0/4}}
\|{\mathbf A}_{N_0} \Psi_{\lambda_m,V}^{N_0} {\bf c}\|_2^2\Big)^{1/2}\nonumber\\
&\ge &
\Big(A_0-\frac{D_1({\mathcal G}) L \alpha}{\alpha-d} N_0^{-\alpha+d}\Big)\|{\bf A}\|_{{\mathcal J}_\alpha({\mathcal G}, {\mathcal V})}
\Big(\sum_{\lambda_m\in G_{N_0/4}} \|\Psi_{\lambda_m,V}^{N_0}{\bf c}\|_2^2\Big)^{1/2}\nonumber\\
 & \ge & \Big(\frac{A_0}{3}-
\frac{D_1({\mathcal G})L\alpha}{3(\alpha-d)} N_0^{-\alpha+d}\Big)\|{\bf A}\|_{{\mathcal J}_\alpha({\mathcal G}, {\mathcal V})}\|{\bf c}\|_2,
\end{eqnarray}
where the last inequality holds
because
for all $i\in V$,
\begin{eqnarray*}
\sum_{\lambda_m\in G_{N_0/4}} |\psi_0(\rho_{\mathcal H}(\lambda_m, i)/N_0)|^2
& \hskip-0.08in \ge & \hskip-0.08in
\Big(\frac{N_0-2}{N_0}\Big)^2
\sum_{\lambda_m\in G_{N_0/4}}\chi_{B_{\mathcal H}(\lambda_m, N_0/2+1)}(i)
\ge \frac{1}{9}
\end{eqnarray*}
by \eqref{psi0.def}, Proposition \ref{covering.pr} and the assumption that $N_0\ge 3$.

Next, we estimate  commutators
$${\mathbf A}_{N_0} \Psi_{\lambda_m,V}^{N_0}-\Psi_{\lambda_m,G}^{N_0} {\mathbf A}_{N_0}=
({\mathbf A}_{N_0} \Psi_{\lambda_m,V}^{N_0}-\Psi_{\lambda_m,G}^{N_0} {\mathbf A}_{N_0})\chi_{\lambda_m,V}^{2N_0},\  \lambda_m\in G_{N_0/4}.$$
Take ${\bf c}=(c(i))_{i\in V}\in \ell^2$. Then
\begin{eqnarray} \label{l2stablility.thm.pf.eq5}
\hskip-0.1in & \hskip-0.1in & \hskip-0.1in \sum_{\lambda_m\in G_{N_0/4}} \|({\mathbf A}_{N_0} \Psi_{\lambda_m,V}^{N_0}-\Psi_{\lambda_m,G}^{N_0} {\mathbf A}_{N_0}){\bf c}\|_2^2
\nonumber \\
\hskip-0.1in &\hskip-0.1in \le & \hskip-0.1in \|{\mathbf A}\|_{{\mathcal J}_\alpha({\mathcal G}, {\mathcal V})}^2
\sum_{\lambda_m\in G_{N_0/4}} \sum_{\lambda\in G}\Big\{\sum_{\rho_{\mathcal H}(\lambda,i)\le  N_0}
(1+\rho_{\mathcal H}(\lambda,i))^{-\alpha} \nonumber\\
\hskip-0.1in &\hskip-0.1in &\hskip-0.1in\times \Big|\psi_0\Big(\frac{\rho_{\mathcal H}(\lambda,\lambda_m)}{N_0}\Big)-\psi_0\Big(\frac{\rho_{\mathcal H}(i,\lambda_m)}{N_0}\Big)\Big|
 \chi_{B_{\mathcal H}(\lambda_m,2N_0)\cap V}(i)
 |c(i)|\Big\}^2\nonumber\\
\hskip-0.1in& \hskip-0.1in\le & \hskip-0.1in 4(D_0({\mathcal G}))^4 N_0^{-2} \|{\mathbf A}\|_{{\mathcal J}_\alpha({\mathcal G}, {\mathcal V})}^2
\Big(\sup_{i\in V} \sum_{\lambda\in B_{\mathcal H}(i, N_0)\cap G} (1+\rho_{\mathcal H}(\lambda,i))^{-\alpha}\rho_{\mathcal H}(\lambda, i)\Big)\nonumber\\
\hskip-0.1in&\hskip-0.1in &\hskip-0.1in
\times \Big(\sup_{\lambda\in G} \sum_{i\in B_{\mathcal H}(\lambda, N_0)\cap V} (1+\rho_{\mathcal H}(\lambda,i))^{-\alpha}\rho_{\mathcal H}(\lambda, i)\Big) \|{\bf c}\|_2^2,
\end{eqnarray}
where the last inequality follows from Propositions \ref{covering.pr} and \ref{shortestpath.pr}, and 
\begin{equation*}
|\psi_{0}(t)-\psi_{0}(t')|\le 2|t-t'| \ \ {\rm for \ all}\ t, t'\in \RR.
\end{equation*}
Following the argument used in \eqref{jaffardpr.pf.eq2}, we have
\begin{eqnarray} \label{l2stablility.thm.pf.eq6}
\hskip-0.08in& \hskip-0.08in& \hskip-0.08in\sup_{i\in V} \sum_{\lambda\in B_{\mathcal H}(i, N_0)\cap G} (1+\rho_{\mathcal H}(\lambda,i))^{-\alpha}\rho_{\mathcal H}(\lambda, i)\nonumber\\
& \hskip-0.08in\le & \hskip-0.08in  \sup_{\lambda'\in G} \sum_{\rho_{\mathcal G}(\lambda, \lambda')\le N_0} (1+\rho_{\mathcal G}(\lambda,\lambda'))^{-\alpha+1}\nonumber\\
& \hskip-0.08in\le & \hskip-0.08in D_1({\mathcal G}) (N_0+1)^{-\alpha+d+1} +(\alpha-1) D_1({\mathcal G}) \sum_{n=0}^{N_0-1}
 (n+1)^{-\alpha+d}\nonumber\\
&\hskip-0.08in \le &\hskip-0.08in D_1({\mathcal G}) (N_0+1)^{-\alpha+d+1} + D_1({\mathcal G})(\alpha-1) \Big(1+\int_1^{N_0} t^{-\alpha+d} dt\Big)\nonumber\\
& \hskip-0.08in\le &\hskip-0.08in \left\{\begin{array}
{ll} \frac{D_1({\mathcal G})(\alpha-1)(\alpha-d)}{\alpha-d-1} & {\rm if} \ \alpha>d+1\\
D_1({\mathcal G})(1+d+ d\ln N_0) & {\rm if} \ \alpha=d+1\\
\frac{2^{d+1-\alpha} D_1({\mathcal G}) d}{d+1-\alpha} N_0^{d+1-\alpha} & {\rm if} \ \alpha<d+1
\end{array}\right. 
\end{eqnarray}
and
\begin{eqnarray} \label{l2stablility.thm.pf.eq7}
\quad & & \sup_{\lambda\in G} \sum_{i\in B_{\mathcal H}(\lambda, N_0)\cap V} (1+\rho_{\mathcal H}(\lambda,i))^{-\alpha}\rho_{\mathcal H}(\lambda, i)\nonumber\\
\quad & \le & L \sup_{\lambda\in G}
\sum_{\lambda'\in B_{\mathcal G}(\lambda, N_0)} (1+\rho_{\mathcal G}(\lambda, \lambda'))^{-\alpha+1}\nonumber\\
\quad & \le & \left\{\begin{array}
{ll} \frac{D_1({\mathcal G}) L (\alpha-1)(\alpha-d)}{\alpha-d-1} & {\rm if} \ \alpha>d+1\\
D_1({\mathcal G})L (1+d+ d\ln N_0) & {\rm if} \ \alpha=d+1\\
\frac{2^{d+1-\alpha} D_1({\mathcal G}) d L }{d+1-\alpha} N_0^{d+1-\alpha} & {\rm if} \ \alpha<d+1.
\end{array}\right.
\end{eqnarray}
Therefore,
\begin{eqnarray*}
  & & (D_0({\mathcal G}))^2 \|{\bf A}_{N_0} {\bf c}\|_2
 \ge
\Big(\sum_{\lambda_m\in G_{N_0/4}}
\|\Psi_{\lambda_m,G}^{N_0} {\mathbf A}_{N_0}{\bf c}\|_2^2\Big)^{1/2}\nonumber\\
&\ge &
\Big(\sum_{\lambda_m\in G_{N_0/4}}
\| {\mathbf A}_{N_0}\Psi_{\lambda_m,V}^{N_0}{\bf c}\|_2^2\Big)^{1/2} -
\Big(\sum_{\lambda_m\in G_{N_0/4}}
\| ({\mathbf A}_{N_0}\Psi_{\lambda_m,V}^{N_0}-\Psi_{\lambda_m,G}^{N_0}{\mathbf A}_{N_0}) {\bf c}\|_2^2\Big)^{1/2}\nonumber\\
& \ge & \frac{A_0\|{\mathbf A}\|_{{\mathcal J}_\alpha({\mathcal G}, {\mathcal V})}}{3}\|{\bf c}\|_2-
D_1({\mathcal G}) L \|{\mathbf A}\|_{{\mathcal J}_\alpha({\mathcal G}, {\mathcal V})}N_0^{-\min(\alpha-d, 1)}\|{\bf c}\|_2\nonumber\\
& & \times  \left\{\begin{array}{ll}  \Big(\frac{ \alpha}{3(\alpha-d)}+\frac{2 (D_0({\mathcal G}))^2(\alpha-1)(\alpha-d)}{\alpha-d-1}\Big) & \hskip-0.1in {\rm if}
\ \alpha>d+1\\
\Big(\frac{d+1}{3}+2(D_0({\mathcal G}))^2(1+d+d\ln N_0)\Big) & \hskip-0.1in {\rm if} \ \alpha=d+1\\
  \Big(\frac{\alpha}{3(\alpha-d)}+\frac{4(D_0({\mathcal G}))^2 d}{d+1-\alpha}\Big) & \hskip-0.1in {\rm if} \ \alpha<d+1,
\end{array}\right.
\end{eqnarray*}
where the first inequality holds by Proposition \ref{covering.pr},
and the third inequality follows from \eqref{l2stablility.thm.pf.eq4}
and \eqref{l2stablility.thm.pf.eq5}.
This together with  Proposition  \ref{bandapproximation.pr}
completes the proof.\end{proof}

\subsection{Proof of Proposition \ref{convergence.prop}} To prove Proposition \ref{convergence.prop}, we need the following critical estimate. 

\begin{pr} \label{wienerbounded.pr}
Let $\mathcal G$, $\mathcal H$, $\mathcal V$ and  ${\bf S}$ be as in Proposition \ref{convergence.prop}.
Then
\begin{equation}
\|(\chi_{\lambda,V}^N {\bf S}^T {\bf S} \chi_{\lambda,V}^N)^{-1}\|_{{\mathcal J}_\alpha({\mathcal V})}
\le \frac{2^{-\alpha-1}(\alpha-d)^2D_2}{\alpha^2D_1D_1(\mathcal{G})\|{\bf S}\|_{{\mathcal J}_\alpha({\mathcal G}, {\mathcal V})}^2},
\end{equation}
where $D_2$ is the  constant in \eqref{convergence.prop.eq4}.
\end{pr}

\begin{proof}
Let ${\bf J}_{\lambda,N}:= \chi_{\lambda,V}^N {\bf S}^T {\bf S}\chi_{\lambda,V}^N$. By Lemma \ref{jaffard.algebra2.lm}, we have
\begin{eqnarray}  \label{convergence.lem.eq-1}
\hskip-0.1in
 \|{\bf J}_{\lambda,N}\|_{{\mathcal J}_\alpha({\mathcal V})}
&\hskip-0.1in \le & \hskip-0.1in \frac{2^{\alpha+1} D_1({\mathcal G})\alpha}{\alpha-d} \|{\bf S}\|_{{\mathcal J}_\alpha({\mathcal G}, {\mathcal V})}^2.
\end{eqnarray}
This together with  Propositions \ref{jaffard.banach.pr} implies that
\begin{equation*}
A^2 \|{\bf S}\|_{{\mathcal J}_\alpha({\mathcal G}, {\mathcal V})}^2 \|\chi_{\lambda,V}^{N}{\bf x}\|_2^2
  \le    \| {\bf S} \chi_{\lambda,V}^{N} {\bf x}\|_2^2  =  \langle {\bf J}_{\lambda,N}{\bf  x},{\bf  x}\rangle
\le  \frac{2^{\alpha+1} \alpha^2 D_1 D_1({\mathcal G})  }{(\alpha-d)^2} \|{\bf S}\|_{{\mathcal J}_\alpha({\mathcal G}, {\mathcal V})}^2
\|\chi_{\lambda,V}^N{\bf x}\|_2^2
\end{equation*}
for all ${\bf x}\in \ell^2$.
Hence
\begin{equation}\label{convergence.lem.pf.eq0}
{\mathbf J}_{\lambda,N}=
\frac{2^{\alpha+1} \alpha^2 D_1 D_1({\mathcal G})  }{(\alpha-d)^2} \|{\bf S}\|_{{\mathcal J}_\alpha({\mathcal G}, {\mathcal V})}^2
({\bf I}_{B_{\mathcal H} (\lambda,N)\cap V} -{\mathbf B}_{\lambda,N})\end{equation}
for some ${\mathbf B}_{\lambda,N}$  satisfying
\begin{equation}  \label{convergence.lem.pf.eq1}\|{\mathbf B}_{\lambda,N}\|_{{\mathcal B}^2}\le  r_0
\end{equation}
and
\begin{equation} \label{convergence.lem.pf.eq2} \|{\mathbf B}_{\lambda,N}\|_{{\mathcal J}_\alpha({\mathcal V})}
 \le
\|{\bf I}_{B_{\mathcal H} (\lambda,N)\cap V}\|_{{\mathcal J}_\alpha({\mathcal V})}+
\frac{2^{-\alpha-1} (\alpha-d)^2\|{\mathbf J}_{\lambda,N}\|_{{\mathcal J}_\alpha(\mathcal{V})}}
{ \alpha^2 D_1 D_1({\mathcal G})   \|{\bf S}\|_{{\mathcal J}_\alpha({\mathcal G}, {\mathcal V})}^2}
 \le  1+ \frac{\alpha-d}{\alpha D_1 }\le 2,
\end{equation}
where ${\bf I}_{B_{\mathcal H} (\lambda,N)\cap V}$ is the identity matrix on $B_{\mathcal H}(\lambda, N)\cap V$.
Then following the argument in \cite{suncasp05} and
applying \eqref{wienerlemma.pf.eq1} with ${\bf C}$ replaced by ${\bf B}_{\lambda,N}$
and $V$ by $B_{\mathcal H}(\lambda, N)\cap V$, we obtain the following estimate
$$\|({\mathbf B}_{\lambda,N})^{n}\|_{{\mathcal J}_{\alpha}(\mathcal{V})}\le
\Big(\frac{ D^{\frac{1}{1-\theta}} \|{\mathbf B}_{\lambda,N}\|_{{\mathcal J}_\alpha(\mathcal{V})}}{\|{\mathbf B}_{\lambda,N}\|_{{\mathcal B}^2}}\Big)^{\frac{2-\theta}{1-\theta}n^{\log_2(2-\theta)}} \|{\mathbf B}_{\lambda,N}\|_{{\mathcal B}^2}^n \ \ {\rm  for\  all}\ n\ge 1,
$$
where
$D=2^{2\alpha+d/2+3} D_1^{1/2} (D_1\alpha/(\alpha-d))^{2-\theta}$.
This together with \eqref{convergence.lem.pf.eq1} and
\eqref{convergence.lem.pf.eq2} leads to
\begin{equation}
\label{convergence.lem.pf.eq3}
\|({\mathbf B}_{k,N})^{n}\|_{{\mathcal J}_{\alpha}(\mathcal{V})}\le
(2D^{\frac{1}{1-\theta}}/r_0)^{\frac{2-\theta}{1-\theta}n^{\log_2(2-\theta)}} r_0^n \ {\rm \ for \ all } \ n\ge 1.
%
\end{equation}
Observe that
\begin{equation} \label{convergence.lem.pf.eq4}
\|({\mathbf J}_{\lambda,N})^{-1}\|_{{\mathcal J}_{\alpha}(\mathcal{V})}
\le
\frac{2^{-\alpha-1}(\alpha-d)^2}{\alpha^2D_1D_1(\mathcal{G}) \|{\bf S}\|_{{\mathcal J}_\alpha(\mathcal{G},\mathcal{V})}^2} \Big(1+\sum_{n=1}^\infty \|({\mathbf B}_{\lambda,N})^n\|_{{\mathcal J}_{\alpha}(\mathcal{V})}\Big)
\end{equation}
by \eqref{convergence.lem.pf.eq0}. Combining \eqref{convergence.lem.pf.eq3}
and \eqref{convergence.lem.pf.eq4} 
 completes the proof.
\end{proof}

\begin{proof}[Proof of Proposition \ref{convergence.prop}]
Observe from \eqref{normal.eqn} and \eqref{localinverse.eqn} that 
\begin{equation*}
\chi_{\lambda,V}^{N/2}({\bf d}_{\lambda, N}-{\bf d}_{2})
 = \chi_{\lambda,V}^{N/2} (\chi_{\lambda,V}^N {\bf S}^T{\bf S}\chi_{\lambda,V}^N)^{-1}
\chi_{\lambda,V}^N {\bf S}^T{\bf S} ({\mathbf I}-\chi_{\lambda,V}^N) {\bf d}_{2}.
\end{equation*}
This together with \eqref{jaffard.banach.pr.pfeq2}, Lemma \ref{jaffard.algebra2.lm}, and Propositions \ref{jaffard.banach.pr} and \ref{wienerbounded.pr} implies that
\begin{eqnarray*}
\|\chi_{\lambda,V}^{N/2}({\bf d}_{\lambda, N}-{\bf d}_{2})\|_\infty
& \le & \|(\chi_{\lambda,V}^N {\bf S}^T {\bf S} \chi_{\lambda,V}^N)^{-1}\chi_{\lambda,V}^N{\bf S}^T{\bf S}\|_{{\mathcal J}_{\alpha}(\mathcal{V})}\times\\
& & \big(\sup \limits_{i\in B_{\mathcal{H}}(\lambda,N/2)\cap V} \sum \limits_{j\notin B_{\mathcal{H}}(\lambda,N)\cap V}(1+\rho_{\mathcal{H}}(i,j))^{-\alpha}\big)\|{\bf d}_{2}\|_\infty\\
& \le & \frac{2^{\alpha+1}D_1\alpha}{\alpha-d} \|(\chi_{\lambda,V}^N {\bf S}^T {\bf S} \chi_{\lambda,V}^N)^{-1}\|_{{\mathcal J}_{\alpha}(\mathcal{V})} \|{\bf S}^T{\bf S}\|_{{\mathcal J}_{\alpha}(\mathcal{V})}\times\\
& & \big(\sup \limits_{i\in  V} \sum \limits_{\rho_{\mathcal{H}}(i,j)>N/2}(1+\rho_{\mathcal{H}}(i,j))^{-\alpha}\big)\|{\bf d}_{2}\|_\infty\\
& \le & 2^{\alpha+1}D_2 \big(\sup \limits_{i\in  V} \sum \limits_{\rho_{\mathcal{H}}(i,j)>N/2}(1+\rho_{\mathcal{H}}(i,j))^{-\alpha}\big)\|{\bf d}_{2}\|_\infty\\
& \le &   \frac{2^{\alpha+1}D_1 D_2\alpha}{\alpha-d}\Big(\frac{N}{2}+1\Big)^{-\alpha+d}\|{\bf d}_{2}\|_\infty\le  D_3(N+1)^{-\alpha+d}\|{\bf d}_{2}\|_\infty.
\end{eqnarray*}
This proves the estimate
\eqref{convergence.prop.eq2}.

Now we prove \eqref{convergence.prop.eq3}.
Set ${\bf y}_{LS}=({\bf S}^T{\bf S})^{-1} {\bf d}_{2}$. By \eqref{jaffard.banach.pr.pfeq2},
\begin{equation}\label{con.1}
\|{\bf y}_{LS}\|_\infty \le \frac{D_1\alpha}{\alpha-d}\|({\bf S}^T{\bf S})^{-1}\|_{{\mathcal J}_\alpha({\mathcal V})}\|{\bf d}_{2}\|_\infty.
\end{equation}
Moreover, following the proof of Proposition \ref{wienerbounded.pr} gives
\begin{equation}\label{con.2}
\|({\bf S}^T {\bf S})^{-1}\|_{{\mathcal J}_\alpha({\mathcal V})}
\le \frac{2^{-\alpha-1}(\alpha-d)^2D_2}{\alpha^2D_1D_1(\mathcal{G})\|{\bf S}\|_{{\mathcal J}_\alpha({\mathcal G}, {\mathcal V})}^2}.
\end{equation}
Write
\begin{eqnarray}\label{con.3}
 \chi_{\lambda,G}^{N/2} ({\bf w}_{\lambda, N}-{\bf w}_{LS})
& = &  \chi_{\lambda,G}^{N/2} (\chi_{\lambda,G}^N {\bf S} \chi_{\lambda,V}^N)
(\chi_{\lambda,V}^N {\bf S}^T{\bf S}\chi_{\lambda,V}^N)^{-2} \chi_{\lambda,V}^N {\bf S}^T{\bf S}(I-\chi_{\lambda,V}^N){\bf d}_{2}\nonumber\\
& & +  \chi_{\lambda,G}^{N/2} (\chi_{\lambda,G}^N {\bf S} \chi_{\lambda,V}^N)
(\chi_{\lambda,V}^N {\bf S}^T{\bf S}\chi_{\lambda,V}^N)^{-1} \chi_{\lambda,V}^N {\bf S}^T{\bf S}(I-\chi_{\lambda,V}^N){\bf y}_{LS}
\nonumber\\
& &  - \chi_{\lambda,G}^{N/2} {\bf S} (I-\chi_{\lambda,V}^N) {\bf y}_{LS}\ \nonumber\\
&=:& {\Rmnum 1}+{\Rmnum 2}+{\Rmnum 3}.
\end{eqnarray}
Using \eqref{bandapproximation.pr.pf.eq1}, \eqref{con.1}, \eqref{con.2}, Lemma \ref{jaffard.algebra2.lm}, and Propositions \ref{jaffard.banach.pr} and \ref{wienerbounded.pr}, we obtain
\begin{eqnarray*}
\|{\Rmnum 1}\|_\infty & \le & \|(\chi_{\lambda,G}^N {\bf S} \chi_{\lambda,V}^N)
(\chi_{\lambda,V}^N {\bf S}^T{\bf S}\chi_{\lambda,V}^N)^{-2} \chi_{\lambda,V}^N {\bf S}^T{\bf S}\|_{{\mathcal J}_\alpha({\mathcal G},{\mathcal V})}\times\\
& & \big(\sup \limits_{\lambda'\in B_{\mathcal{H}}(\lambda,N/2)\cap G} \sum \limits_{i\notin B_{\mathcal{H}}(\lambda,N)\cap V}(1+\rho_{\mathcal{H}}(\lambda',i))^{-\alpha}\big)\|{\bf d}_{2}\|_\infty\\
& \le & \frac{2^{2\alpha+2}LD_2^2}{\|\mathbf{S}\|_{{\mathcal J}_\alpha({\mathcal G},{\mathcal V})}}\big(\sup \limits_{\lambda'\in G} \sum \limits_{ \rho_{\mathcal{H}}(\lambda',i)>N/2}(1+\rho_{\mathcal{H}}(\lambda',i))^{-\alpha}\big)\|{\bf d}_{2}\|_\infty\\
& \le & \frac{2^{3\alpha-d+2}\alpha L^2D_1(\mathcal{G})D_2^2}{(\alpha-d)\|\mathbf{S}\|_{{\mathcal J}_\alpha({\mathcal G},{\mathcal V})}}(N+1)^{-\alpha+d}\|{\bf d}_{2}\|_\infty,
\end{eqnarray*}
\begin{eqnarray*}
\|{\Rmnum 2}\|_\infty & \le & \frac{2^{3\alpha-d+2}\alpha^2 L^2(D_1(\mathcal{G}))^2D_2}{(\alpha-d)^2}\|\mathbf{S}
\|_{\mathcal{J}_\alpha(\mathcal{G},\mathcal{V})}
(N+1)^{-\alpha+d}\|{\bf y}_{LS}\|_\infty\\
& \le & \frac{2^{2\alpha-d+1}\alpha L^2D_1(\mathcal{G})D_2^2}{(\alpha-d)\|\mathbf{S}\|_{{\mathcal J}_\alpha({\mathcal G},{\mathcal V})}}(N+1)^{-\alpha+d}\|{\bf d}_{2}\|_\infty,
\end{eqnarray*} and
\begin{eqnarray*}
\|{\Rmnum 3}\|_\infty & \le & \frac{LD_2}{\|\mathbf{S}\|_{{\mathcal J}_\alpha({\mathcal G},{\mathcal V})}}(N+1)^{-\alpha+d}\|{\bf d}_{2}\|_\infty.
\end{eqnarray*}These together with \eqref{con.3} prove \eqref{convergence.prop.eq3}.
\end{proof}

\subsection{Proof of Theorem \ref{convergence.thm}}
Let
\begin{equation}\label{un.def}{\bf u}_n={\bf S}^T({\bf w}_{n}-{\bf w}_{LS})={\bf S}^T {\bf w}_{n}-{\bf d}_{2} \ \ {\rm and} \ \ {\bf v}_n={\bf S} {\bf u}_n,  \ n\ge 1.\end{equation}
Then,
\begin{equation*}
{\bf u}_{n+1} = {\bf u}_n-{\bf S}^T {\bf R}_N {\bf S}^T{\bf S}
{\bf u}_n  =  {\bf S}^T \big({\bf S} ({\bf S}^T{\bf S})^{-2} {\bf S}^T {\bf v}_n-{\bf R}_N {\bf S}^T {\bf v}_n\big)
\end{equation*}
by \eqref{wn.def0}, \eqref{wn.def} and \eqref{un.def}. Therefore,
\begin{eqnarray}\label{con.5}
\|{\bf u}_{n+1}\|_\infty&\le &  \frac{D_1(\mathcal{G})L\alpha}{\alpha-d} \|{\bf S}\|_{{\mathcal J}_\alpha(\mathcal{G},\mathcal{V})}\|{\bf R}_N {\bf S}^T {\bf v}_n- {\bf S} ({\bf S}^T{\bf S})^{-2} {\bf S}^T {\bf v}_n\|_\infty\nonumber\\
& \le  &  \frac{D_1(\mathcal{G})D_4L\alpha}{\alpha-d} \|{\bf S}\|_{{\mathcal J}_\alpha(\mathcal{G},\mathcal{V})} (N+1)^{-\alpha+d} \|({\bf S}^T{\bf S})^{-1} {\bf S}^T {\bf v}_n\|_\infty\nonumber\\
&  = &  r_1 \|{\bf u}_{n}\|_\infty\le \cdots\le
r_1^n\|\mathbf{S}^T({\bf R}_N {\bf S}^T{\bf S}-{\bf S}({\bf S}^T{\bf S})^{-1}){\bf d}_{2}\|_\infty \nonumber\\
&\le &  r_1^{n+1}\|{\bf d}_{2}\|_\infty,
  \end{eqnarray}
where the second inequality follows from \eqref{rxerror} with ${\bf d}_{2}$ replaced by $({\bf S}^T{\bf S})^{-1} {\bf S}^T {\bf v}_n$,
and the  last inequality holds by
\eqref{rxerror} and Proposition \ref{jaffard.pr}.

Observe that
\begin{equation}
{\bf w}_{n+1}-{\bf w}_n=- {\bf R}_N {\bf S}^T{\bf S} {\bf u}_n.
\end{equation}
Using \eqref{con.4}, Proposition \ref{jaffard.pr} and Lemma \ref{jaffard.algebra2.lm} gives
\begin{equation}
\|{\bf w}_{n+1}-{\bf w}_n\|_\infty \leq
 \frac{2^{2\alpha+2}\alpha L^3(D_1(\mathcal{G}))^2D_2^2}{(\alpha-d)D_1\|{\bf S}\|_{{\mathcal J}_\alpha(\mathcal{G},\mathcal{V})}}\|{\bf u}_n\|_\infty.
\end{equation}
This together with \eqref{con.5} proves the exponential convergence \eqref{convergence.thm.eq1}.

 The conclusion \eqref{con.6}  follows from \eqref{un.def} by taking limit $n \to \infty$.

 The error estimate \eqref{convergence.thm.eq2} between the ``least square'' solution ${\bf d}_2$ and its sub-optimal approximation ${\bf S}^T{\bf w}_n, n\ge 1$, follows from \eqref{convergence.thm.eq1} and Proposition \ref{jaffard.pr}.

\end{document}